\documentclass[11pt]{article}%
\usepackage{amsmath}
\usepackage{amsfonts}
\usepackage{amssymb}
\usepackage{graphicx}
\usepackage{mathtools}
\usepackage{epstopdf}
\usepackage{color}
\usepackage{float}
\usepackage{subfig}%
\providecommand{\U}[1]{\protect\rule{.1in}{.1in}}
\newtheorem{theorem}{Theorem}[section]

\newtheorem{condition}{Assumption}[section]

\newtheorem{corollary}{Corollary}[section]

\newtheorem{example}{Example}[section]

\newtheorem{lemma}{Lemma}[section]

\newtheorem{proposition}{Proposition}[section]
\newtheorem{remark}{Remark}[section]

\newenvironment{proof}[1][Proof]{\noindent\textbf{#1.} }{\ \rule{0.5em}{0.5em}}
\numberwithin{equation}{section}
\newcommand{\pr}{\mathbb{P}}

\newcommand{\HZ}{\textcolor[rgb]{0.0,0.0,0.0}}
\newcommand{\BIN}{\textcolor[rgb]{0,0,0}}

\begin{document}

\title{\textcolor[rgb]{0.0,0.0,0.0}{A} Unified Approach for Drawdown (Drawup) of
Time-Homogeneous Markov Processes}
\author{David Landriault\thanks{Department of Statistics and Actuarial Science,
University of Waterloo, Waterloo, ON, N2L 3G1, Canada (dlandria@uwaterloo.ca)}
\and Bin Li\thanks{Corresponding Author: Department of Statistics and Actuarial
Science, University of Waterloo, Waterloo, ON, N2L 3G1, Canada
(bin.li@uwaterloo.ca)}
\and Hongzhong Zhang\thanks{Department of IEOR, Columbia University, New York, NY,
10027, USA (hz2244@columbia.edu)}}
\date{{\small \today}}
\maketitle

\begin{abstract}
Drawdown (resp. drawup) of a stochastic process, also referred as the
reflected process at its supremum (resp. infimum), has wide applications in
many areas including financial risk management, actuarial mathematics and
statistics. In this paper, for general time-homogeneous Markov processes, we
study the joint law of the first passage time of the drawdown (resp. drawup)
process, its overshoot, and the maximum of the underlying process at this
first passage time. By using short-time pathwise analysis, under some mild
regularity conditions, the joint law of the three drawdown quantities is shown
to be the unique solution to an integral equation which is expressed in terms
of fundamental two-sided exit quantities of the underlying process. Explicit
forms for this joint law are found when the Markov process has only one-sided
jumps or is a L\'{e}vy process (possibly with two-sided jumps). The proposed
methodology provides a unified approach to
\HZ{study various drawdown quantities} for the general
class of time-homogeneous Markov processes.

\textit{Keywords}: Drawdown; Integral equation; Reflected process;
Time-homogeneous Markov process

\textit{MSC}(2000): Primary 60G07; Secondary 60G40

\end{abstract}

\baselineskip15.5pt

\section{Introduction}

We consider a time-homogeneous, real-valued, non-explosive, \BIN{c\`{a}dl\`{a}g}
Markov process $X=(X_{t})_{t\geq0}$ with state space $\mathbb{R}$
\footnote{The state space can \HZ{sometimes be relaxed} to
an open interval of $%
\mathbb{R}
$
\HZ{(e.g., (0,+$\infty$) for geometric Brownian motions)}.
It is also possible to treat some general state space with complex boundary
behaviors. However, for simplicity, we choose $%
\mathbb{R}
$ as the state space of $X$ in this paper.} defined on a filtered probability
space $(\Omega,\mathcal{F},\boldsymbol{F}=(\mathcal{F}_{t})_{t\geq
0},\mathbb{P})$
\textcolor[rgb]{0.0,0.0,0.0}{with a complete and right-continuous filtration.}
\HZ{Throughout, we silently assume that $X$ satisfies the strong Markov property (see Section III.8,9 of Rogers and Williams \cite{RW00}), and exclude Markov processes with monotone paths.}
The first passage time of $X$ above (below) a level $x\in%
\mathbb{R}
$ is denoted by
\[
\textcolor[rgb]{0.0,0.0,0.0}{T_{x}^{+(-)}=\inf\left\{  t\geq0:X_{t}>(<)x\right\},}
\]
with the common convention that $\inf\emptyset=\infty$.

The drawdown process of $X$ (also known as the reflected process of $X$ at its
supremum) is denoted by $Y=(Y_{t})_{t\geq0}$ with $Y_{t}=M_{t}-X_{t},$ where
$M_{t}=\sup_{0\leq s\leq t}X_{t}$. Let $\tau_{a}=\inf\{t>0:Y_{t}>a\}$ be the
first time the magnitude of drawdowns exceeds a given threshold $a>0$. Note
that $\left(  \textcolor[rgb]{0.0,0.0,0.0}{\sup_{0\leq s\leq t}}Y_{s}%
>a\right)  =\left(  \tau_{a}\leq t\right)  $ $\mathbb{P}$-a.s. Hence, the
distributional study of the maximum drawdown of $X$ is equivalent to the study
of the stopping time $\tau_{a}$. Similarly, the drawup process of $X$ is
defined as $\hat{Y}_{t}=X_{t}-m_{t}$ for $t\geq0,$ where $m_{t}=\inf_{0\leq
s\leq t}X_{t}$. However, given that the drawup of $X$ can be investigated via
the drawdown of $-X$, we exclusively focus on the drawdown process $Y$ in this paper.

Applications of drawdowns can be found in many areas. For instance, drawdowns
are widely used by mutual funds and commodity trading advisers to quantify
downside risks. Interested readers are referred to Schuhmacher and Eling
\cite{SE11} for a review of drawdown-based performance measures. An extensive
body of literature exists on the assessment and mitigation of drawdown risks
(e.g., Grossman and Zhou \cite{GZ93}, Carr et al. \cite{CZH11}, Cherny and
Obloj \cite{CO13}, and Zhang et al. \cite{ZLH13}). Drawdowns are also closely
related to many problems in mathematical finance, actuarial science and
statistics such as the pricing of Russian options (e.g., Shepp and Shiryaev
\cite{SS93}, Asmussen et al. \cite{AAP04} and Avram et al. \cite{AKP04}),
\BIN{De Finetti's} dividend problem\ (e.g., Avram et al.
\cite{APP07} and Loeffen \cite{L08}), loss-carry-forward taxation models
(e.g., Kyprianou and Zhou \cite{KZ09} and Li et al. \cite{LTZ13}), and
change-point detection methods (e.g., Poor and Hadjiliadis \cite{PH09}). \BIN{More specifically, in De Finetti's dividend problem under a fixed dividend barrier $a>0$, the underlying surplus process with dividend payments is a process obtained from reflecting $X$ at a fixed barrier $a$ (the reflected process' dynamics may be different than the drawdown process $Y$ when the underlying process $X$ is not spatial homogeneous). However, the distributional study of ruin quantities in De Finetti's dividend problem can be transformed to the study of drawdown quantities for the underlying surplus process; see Kyprianou and Palmowski \cite{KP07} for a more detailed discussion. Similarly, ruin problems in loss-carry-forward taxation models can also be transformed to a generalized drawdown problem for classical models without taxation, where the generalized drawdown process is defined in the form of $Y_t=\gamma(M_t)-X_t$ for some measurable function $\gamma(\cdot)$.}

The distributional study of drawdown quantities is not only of theoretical
interest, but also plays a fundamental role in the aforementioned
applications. Early distributional studies on drawdowns date back to Taylor
\cite{T75} on the joint Laplace transform of $\tau_{a}$ and $M_{\tau_{a}}$ for
Brownian motions. This result was later generalized by Lehoczky \cite{L77} to
time-homogeneous diffusion processes. Douady et al. \cite{DSY00} and Magdon et
al. \cite{MAPA04} derived infinite series expansions for the distribution of
$\tau_{a}$ for a standard Brownian motion and a drifted Brownian motion,
respectively. For spectrally negative L\'{e}vy processes, Mijatovic and
Pistorius \cite{MP12} obtained a sextuple formula for the joint Laplace
transform of $\tau_{a}$ and the last reset time of the maximum prior to
$\tau_{a}$, together with the joint distribution of the running maximum, the
running minimum, and the overshoot of $Y$ at $\tau_{a}$. For some studies on
the joint law of drawdown and drawup of spectrally negative L\'{e}vy processes
or diffusion processes, please refer to Pistorius \cite{P04}, Pospisil et al.
\cite{PVH09}, Zhang and Hadjiliadis \cite{ZH10}, and Zhang \cite{Z15}.

As mentioned above, L\'{e}vy processes\footnote{Most often, one-sided L\'{e}vy
processes (an exception to this is Baurdoux \cite{B09} for general L\'{e}vy
processes)} and time-homogeneous diffusion processes are two main classes of
Markov processes for which various drawdown problems have been extensively
studied. The treatment of these two classes of Markov processes has typically
been considered distinctly in the literature. For L\'{e}vy processes,
It\^{o}'s excursion theory is a powerful approach to handle drawdown problems
(e.g., Avram et al. \cite{AKP04}, Pistorius \cite{P04}, and Mijatovic and
Pistorius \cite{MP12}). However, the excursion-theoretic approach is somewhat
specific to the underlying model, and additional care is required when a more
general class of Markov processes is considered. On the other hand, for
time-homogeneous diffusion processes, Lehoczky \cite{L77} introduced an
ingenious approach which has recently been generalized by many researchers
(e.g., Zhou \cite{Z07}, Li et al. \cite{LTZ13}, and Zhang \cite{Z15}). Here
again, Lehoczky's approach relies on the continuity of the sample path of the
underlying model, and hence is not applicable for processes with upward jumps.
Also, other general methodologies (such as the martingale approach in, e.g.,
Asmussen \cite{AAP04} and the occupation density approach in, e.g., Ivanovs
and Palmowski \cite{IP12}) are well documented in the literature but they
strongly depend on the specific structure of the underlying process. To the
best of our knowledge, no unified treatment of drawdowns (drawups) for general
Markov processes has been proposed in the literature.

In this paper, we propose a general and unified approach to study the joint
law of $(\tau_{a},M_{\tau_{a}},Y_{\tau_{a}})$ for time-homogeneous Markov
processes with possibly two-sided jumps. Under mild regularity conditions, the
joint law is expressed as the solution to an integral equation which involves
two-sided exit quantities of the underlying process $X$. The uniqueness of the
integral equation for the joint law is also investigated. In particular, the
joint law possesses explicit forms when $X$ has only one-sided jumps or is a
L\'{e}vy process (possibly with two-sided jumps). In general, our main result
reduces the drawdown problem to fundamental two-sided exit quantities.

The main idea of our proposed approach is briefly summarized below. By
analyzing the evolution of sample paths over a short time period following
time $0$ and using renewal arguments, we first establish tight upper and lower
bounds for the joint law of $(\tau_{a},M_{\tau_{a}},Y_{\tau_{a}})$ in terms of
the two-sided exit quantities. Then, under mild regularity conditions, we use
a Fatou's lemma with varying measures to show that the upper and lower bounds
converge when the length of the time interval approaches $0$. This leads to an
integro-differential equation satisfied by the desired joint law. Finally, we
reduce the integro-differential equation to an integral equation. When $X$ is
a spectrally negative Markov process or a general L\'{e}vy process, the
integral equation can be solved and the joint law of $(\tau_{a},M_{\tau_{a}%
},Y_{\tau_{a}})$ is hence explicitly expressed in terms of two-sided exit quantities.

The rest of the paper is organized as follows. In Section 2, we introduce some
fundamental two-sided exit quantities and present several preliminary results.
In Section 3, we derive the joint law of $(\tau_{a},Y_{\tau_{a}},M_{\tau_{a}%
})$ for general time-homogeneous Markov processes. Several Markov processes
for which the proposed regularity conditions are met are further discussed.
\BIN{Some numerical examples are investigated in more detail in Section 4}.
Some technical proofs are postponed to Appendix.

\section{Preliminary}

For ease of notation, we adopt the following conventions throughout the paper.
We denote by $\mathbb{P}_{x}$ the law of $X$ given $X_{0}=x\in%
\mathbb{R}
$ and write $\mathbb{P}\equiv\mathbb{P}_{0}$ for brevity. We write $u\wedge
v=\min\{u,v\}$, $%
\mathbb{R}
_{+}=[0,\infty)$, and $\int_{x}^{y}\cdot\mathrm{d}z$ for an integral on the
open interval $z\in(x,y)$.

For $q,s\geq0$, $u\leq x\leq v$ and $z>0$, we introduce the following
two-sided exit quantities of $X$:
\begin{align*}
B_{1}^{(q)}(x;u,v)  &  :=\mathbb{E}_{x}\left[  e^{-qT_{v}^{+}}1_{\left\{T_{v}^{+}<\infty,  
T_{v}^{+}<T_{u}^{-},X_{T_{v}^{+}}=v\right\}  }\right]  ,\\
B_{2}^{(q)}(x,\mathrm{d}z;u,v)  &  :=\mathbb{E}_{x}\left[  e^{-qT_{v}^{+}%
}1_{\left\{T_{v}^{+}<\infty,    T_{v}^{+}<T_{u}^{-},X_{T_{v}^{+}}-v\in\mathrm{d}z\right\}
}\right]  ,\\
C^{(q,s)}(x;u,v)  &  :=\mathbb{E}_{x}\left[  e^{-qT_{u}^{-}-s(u-X_{T_{u}^{-}%
})}1_{\left\{ T_{u}^{-}<\infty,   T_{u}^{-}<T_{v}^{+}\right\}  }\right]  .
\end{align*}
We also define the joint Laplace transform
\begin{equation}
B^{(q,s)}(x;u,v):=\mathbb{E}_{x}\left[  e^{-qT_{v}^{+}-s(X_{T_{v}^{+}}%
-v)}1_{\left\{ T_{v}^{+}<\infty,  T_{v}^{+}<T_{u}^{-}\right\}  }\right]  =B_{1}^{(q)}%
(x;u,v)+B_{2}^{(q,s)}(x;u,v), \label{BBB}%
\end{equation}
where $B_{2}^{(q,s)}(x;u,v):=\int_{0}^{\infty}e^{-sz}B_{2}^{(q)}%
(x,\mathrm{d}z;u,v)$.

The following pathwise inequalities are central to the construction of tight
bounds for the joint law of the triplet $(\tau_{a},M_{\tau_{a}},Y_{\tau_{a}})$.

\begin{proposition}
\label{prop path}For $q,s\geq0$, $x\in\mathbb{R}$ and $\varepsilon\in(0,a)$,
we have $\mathbb{P}_{x}$-a.s.
\begin{equation}
1_{\{T_{x+\varepsilon}^{+}<\infty, T_{x+\varepsilon}^{+}<T_{x+\varepsilon-a}^{-}\}}\leq1_{\{T_{x+\varepsilon}^{+}<\infty, T_{x+\varepsilon
}^{+}<\tau_{a}\}}\leq1_{\{T_{x+\varepsilon}^{+}<\infty, T_{x+\varepsilon}^{+}<T_{x-a}^{-}\}}, \label{eq.up}%
\end{equation}
and
\begin{align}
e^{-q\tau_{a}-s(Y_{\tau_{a}}-a)}1_{\left\{\tau_{a}<\infty,  \tau_{a}<T_{x+\varepsilon}%
^{+}\right\}  }  &  \geq e^{-qT_{x-a}^{-}-s(x-a-X_{T_{x-a}^{-}})-s\varepsilon
}1_{\{T_{x-a}^{-}<\infty, T_{x-a}^{-}<T_{x+\varepsilon}^{+}\}},\label{eq.down1}\\
e^{-q\tau_{a}-s(Y_{\tau_{a}}-a)}1_{\left\{ \tau_{a}<\infty, \tau_{a}<T_{x+\varepsilon}%
^{+}\right\}  }  &  \leq e^{-qT_{x+\varepsilon-a}^{-}%
-s(x-a-X_{T_{x+\varepsilon-a}^{-}})}1_{\{T_{x+\varepsilon-a}^{-}<\infty, T_{x+\varepsilon-a}^{-}%
<T_{x+\varepsilon}^{+}\}}. \label{eq.down2}%
\end{align}

\end{proposition}

\begin{proof}
By analyzing the sample paths of $X$, it is easy to see that, for any path $\omega\in(T_{x+\varepsilon}^{+}<\infty)$, we have $\pr_x\{\tau_a\le T_{x-a}^-\}=1$, so
\begin{equation}
(T_{x+\varepsilon}^{+}<\infty, T_{x+\varepsilon}^{+}<\tau_{a})=(T_{x+\varepsilon}^{+}<\infty, T_{x+\varepsilon
}^{+}<\tau_a\le T_{x-a}^{-})\subset(T_{x+\varepsilon}^{+}<\infty,  T_{x+\varepsilon}^{+}<T_{x-a}^-)\quad\mathbb{P}_{x}\text{-a.s.}\nonumber
\label{w1}%
\end{equation}
and similarly, $\pr_x$-a.s.%
\begin{equation}
(T_{x+\varepsilon}^{+}<\infty, T_{x+\varepsilon}^{+}<T_{x+\varepsilon-a}^-)=(T_{x+\varepsilon}^{+}<\infty, T_{x+\varepsilon
}^{+}< T_{x+\varepsilon-a}^{-}, T_{x+\varepsilon
}^{+}< \tau_a)\subset(T_{x+\varepsilon}^{+}<\infty,  T_{x+\varepsilon}^{+}<\tau_a),\nonumber%
\end{equation}
which immediately implies (\ref{eq.up}). On the other hand, by using the same argument, we have
\begin{equation}
(T_{x-a}^{-}<\infty, T_{x-a}^{-}<T_{x+\varepsilon}^{+})=(T_{x-a}^{-}<\infty, \tau_{a}\leq T_{x-a}^{-}<T_{x+\varepsilon
}^{+})\subset(\tau_{a}<\infty, \tau_{a}<T_{x+\varepsilon}^{+})\quad\mathbb{P}_{x}\text{-a.s.}
\label{w1}%
\end{equation}
and%
\begin{equation}
(\tau_{a}<\infty, \tau_{a}<T_{x+\varepsilon}^{+})=(\tau_{a}<\infty, T_{x+\varepsilon-a}^{-}\leq\tau
_{a}<T_{x+\varepsilon}^{+})\subset(T_{x+\varepsilon-a}^{-}<\infty, T_{x+\varepsilon-a}^{-}<T_{x+\varepsilon
}^{+})\quad\mathbb{P}_{x}\text{-a.s.} \label{w2}%
\end{equation}
For any path $\omega\in(T_{x-a}^{-}<\infty, T_{x-a}^{-}<T_{x+\varepsilon}^{+})$, we know
from (\ref{w1}) that $\omega\in(T_{x-a}^{-}<\infty, \tau_{a}\leq T_{x-a}^{-}<T_{x+\varepsilon}%
^{+})$. This implies $M_{\tau_{a}}(\omega)\leq x+\varepsilon$ and $X_{\tau
_{a}}(\omega)\geq X_{T_{x-a}^{-}}(\omega)$, which further entails that
$Y_{\tau_{a}}(\omega)=M_{\tau_{a}}(\omega)-X_{\tau_{a}}(\omega)\leq
x+\varepsilon-X_{T_{x-a}^{-}}(\omega)$. Therefore, by the above analysis and
the second inequality of (\ref{eq.up}),
\[
e^{-qT_{x-a}^{-}-s(x+\varepsilon-X_{T_{x-a}^{-}})}1_{\left\{ T_{x-a}^{-}<\infty, T_{x-a}%
^{-}<T_{x+\varepsilon}^{+}\right\}  }\leq e^{-q\tau_{a}-sY_{\tau_{a}}%
}1_{\left\{\tau_{a}<\infty,  \tau_{a}<T_{x+\varepsilon}^{+}\right\}  }\quad\mathbb{P}%
_{x}\text{-a.s.}%
\]
which naturally leads to (\ref{eq.down1}).

Similarly, for any sample path $\omega\in(\tau_{a}<\infty, \tau_{a}<T_{x+\varepsilon}^{+})$,
we know from (\ref{w2}) that $\omega\in(\tau_{a}<\infty, T_{x+\varepsilon-a}^{-}\leq\tau
_{a}<T_{x+\varepsilon}^{+})$, which implies that $x-X_{T_{x+\varepsilon
-a}^{-}}(\omega)\leq Y_{T_{x+\varepsilon-a}^{-}}(\omega)\leq Y_{\tau_{a}%
}(\omega).$ Therefore, by the first inequality of (\ref{eq.up}), we obtain
\[
e^{-q\tau_{a}-sY_{\tau_{a}}}1_{\left\{ \tau_{a}<\infty, \tau_{a}<T_{x+\varepsilon}%
^{+}\right\}  }\leq e^{-qT_{x+\varepsilon-a}^{-}-s(x-X_{T_{x+\varepsilon
-a}^{-}})}1_{\{T_{x+\varepsilon-a}^{-}<\infty, T_{x+\varepsilon-a}^{-}<T_{x+\varepsilon}^{+}\}}\quad
\mathbb{P}_{x}\text{-a.s.}%
\]
This implies the second inequality of (\ref{eq.down2}).\bigskip
\end{proof}

By Proposition \ref{prop path}, we easily obtain the following useful estimates.

\begin{corollary}
\label{cor bd}For $q,s\geq0$, $x\in%
\mathbb{R}
,z>0$ and $\varepsilon\in(0,a)$,%
\begin{align*}
B_{1}^{(q)}(x;x+\varepsilon-a,x+\varepsilon)  &  \leq\mathbb{E}_{x}\left[
e^{-qT_{x+\varepsilon}^{+}}1_{\{T_{x+\varepsilon}^{+}<\infty, T_{x+\varepsilon}^{+}<\tau_{a}%
,X_{T_{x+\varepsilon}^{+}}=x+\varepsilon\}}\right]  \leq B_{1}^{(q)}%
(x;x-a,x+\varepsilon),\\
B_{2}^{(q)}(x,\mathrm{d}z;x+\varepsilon-a,x+\varepsilon)  &  \leq
\mathbb{E}_{x}\left[  e^{-qT_{x+\varepsilon}^{+}}1_{\{T_{x+\varepsilon}^{+}<\infty, T_{x+\varepsilon}%
^{+}<\tau_{a},X_{T_{x+\varepsilon}^{+}}-x-\varepsilon\in\mathrm{d}z\}}\right]
\leq B_{2}^{(q)}(x,\mathrm{d}z;x-a,x+\varepsilon),
\end{align*}
and%
\[
e^{-s\varepsilon}C^{(q,s)}(x;x-a,x+\varepsilon)\leq\mathbb{E}_{x}\left[
e^{-q\tau_{a}-s(Y_{\tau_{a}}-a)}1_{\left\{\tau_{a}<\infty,  \tau_{a}<T_{x+\varepsilon}%
^{+}\right\}  }\right]  \leq e^{s\varepsilon}C^{(q,s)}(x;x+\varepsilon
-a,x+\varepsilon).
\]

\end{corollary}

\begin{remark}
\normalfont It is not difficult to check that the results of Proposition
\ref{prop path} and Corollary \ref{cor bd} still hold if the first passage
times and the drawdown times are only observed discretely or randomly (such as
the Poisson observation framework in Albrecher et al. \cite{AIZ16} for the
latter). Further, explicit relationship between Poisson observed first passage
times and Poisson observed drawdown times (similar as for Theorem
\ref{thm markov} below) can be found by exploiting the same approach as laid
out in this paper.
\end{remark}

The later analysis involves the weak convergence of measures which is recalled
here. Consider a metric space $S$ with the Borel $\sigma$-algebra on it. We
say a sequence of finite measures $\{\mu_{n}\}_{n\in%
\mathbb{N}
}$ is weakly convergent to a finite measure $\mu$ as $n\rightarrow\infty$ if
\[
\lim_{n\rightarrow\infty}\int_{S}\phi(z)\mathrm{d}\mu_{n}(z)=\int_{S}%
\phi(z)\mathrm{d}\mu(z),
\]
for any bounded and continuous function $\phi(\cdot)$ on $S$.

In the next lemma, we show some forms of Fatou's lemma for varying measures
under weak convergence. Similar results are proved in Feinberg et al.
\cite{FKZ14} for probability measures. For completeness, a proof for general
finite measures is provided in Appendix.

\begin{lemma}
\label{lem fatou}Suppose that $\{\mu_{n}\}_{n\in%
\mathbb{N}
}$ is a sequence of finite measures on $S$ which is weakly convergent to a
finite measure $\mu$, and $\{\phi_{n}\}_{n\in%
\mathbb{N}
}$ is a sequence of uniformly bounded and nonnegative functions on $S$. Then,%
\begin{equation}
\int_{S}\liminf_{n\rightarrow\infty,w\rightarrow z}\phi_{n}(w)\mathrm{d}%
\mu(z)\leq\liminf_{n\rightarrow\infty}\int_{S}\phi_{n}(z)\mathrm{d}\mu
_{n}(z)\text{,} \label{inf}%
\end{equation}
and
\begin{equation}
\int_{S}\limsup_{n\rightarrow\infty,w\rightarrow z}\phi_{n}(w)\mathrm{d}%
\mu(z)\geq\limsup_{n\rightarrow\infty}\int_{S}\phi_{n}(z)\mathrm{d}\mu_{n}(z).
\label{sup}%
\end{equation}

\end{lemma}

\section{Main results}

In this section, we study the joint law of $(\tau_{a},M_{\tau_{a}},Y_{\tau
_{a}})$ for a general Markov process with possibly two-sided jumps. The
following assumptions on the two-sided exit quantities of $X$ are assumed to
hold, which are sufficient (but not necessary) conditions for the
applicability of our proposed methodology. Weaker assumptions might be assumed
for special Markov processes; see, for instance, Remark \ref{rk levy} and
Corollary \ref{cor snm} below.

\begin{condition}
For all $q,s\geq0$, $z>0$ and $x>X_{0}$, we assume the following limits exist
and identities hold:
\begin{align*}
\text{\textbf{(A1)} }b_{a,1}^{(q)}(x)  &  :=\lim_{\varepsilon\downarrow0}%
\frac{1-B_{1}^{(q)}(x;x-a,x+\varepsilon)}{\varepsilon}=\lim_{\varepsilon
\downarrow0}\frac{1-B_{1}^{(q)}(x;x+\varepsilon-a,x+\varepsilon)}{\varepsilon
}\\
&  =\lim_{\varepsilon\downarrow0}\frac{1-B_{1}^{(q)}(x-\varepsilon
;x-a,x)}{\varepsilon}=\lim_{\varepsilon\downarrow0}\frac{1-B_{1}%
^{(q)}(x-\varepsilon;x-\varepsilon-a,x)}{\varepsilon},
\end{align*}
and $\int_{x}^{y}b_{a,1}^{(q)}(w)\mathrm{d}w<\infty$ for any $x,y\in%
\mathbb{R}
$;%
\begin{align*}
\text{\textbf{(A2)} }b_{a,2}^{(q,s)}(x)  &  :=\lim_{\varepsilon\downarrow
0}\frac{1}{\varepsilon}B_{2}^{(q,s)}(x;x-a,x+\varepsilon)=\lim_{\varepsilon
\downarrow0}\frac{1}{\varepsilon}B_{2}^{(q,s)}(x;x+\varepsilon-a,x+\varepsilon
)\\
&  =\lim_{\varepsilon\downarrow0}\frac{1}{\varepsilon}B_{2}^{(q,s)}%
(x-\varepsilon;x-a,x)=\lim_{\varepsilon\downarrow0}\frac{1}{\varepsilon}%
B_{2}^{(q,s)}(x-\varepsilon;x-\varepsilon-a,x),
\end{align*}
and $s\longmapsto b_{a,2}^{(q,s)}(x)$ is right continuous at $s=0$;%
\begin{align*}
\text{\textbf{(A3)} }c_{a}^{(q,s)}(x)  &  :=\lim_{\varepsilon\downarrow0}%
\frac{C^{(q,s)}(x;x-a,x+\varepsilon)}{\varepsilon}=\lim_{\varepsilon
\downarrow0}\frac{C^{(q,s)}(x;x+\varepsilon-a,x+\varepsilon)}{\varepsilon}\\
&  =\lim_{\varepsilon\downarrow0}\frac{C^{(q,s)}(x-\varepsilon;x-a,x)}%
{\varepsilon}=\lim_{\varepsilon\downarrow0}\frac{C^{(q,s)}(x-\varepsilon
;x-\varepsilon-a,x)}{\varepsilon}.
\end{align*}

\end{condition}

Under Assumptions (\textbf{A1}) and (\textbf{A2}), it follows from (\ref{BBB})
that
\begin{equation}
b_{a}^{(q,s)}(x):=\lim_{\varepsilon\downarrow0}\frac{1-B^{(q,s)}%
(x;x-a,x+\varepsilon)}{\varepsilon}=b_{a,1}^{(q)}(x)-b_{a,2}^{(q,s)}(x).
\label{bbb}%
\end{equation}

\begin{remark}
\label{rmk31} \normalfont Due to the general structure of $X$, it is difficult
to refine Assumptions \textbf{(A1)}-\textbf{(A3)} unless a specific structure
for $X$ is given. \BIN{A necessary condition for
Assumptions \textbf{(A1)}-\textbf{(A3}) to hold is that,
\[
T_{x}^{+}=0\text{ and }X_{T_{x}^{+}}=x,\text{ }\mathbb{P}_{x}\text{-a.s. for
all }x\in\mathbb{R}\text{.}\]
In other words, $X$ must be upward regular and creeping upward at every
$x$.}\footnote{See page 142 and page 197 of \cite{K14} for definitions of
regularity and creeping for L\'{e}vy processes.} In the later part of this
section, we provide some examples of Markov processes which satisfy
Assumptions \textbf{(A1)}-\textbf{(A3)}, including spectrally negative
L\'{e}vy processes, linear diffusions, piecewise exponential Markov processes,
and jump diffusions.
\end{remark}

\begin{remark}
\normalfont\label{rk weak}By Theorem 5.22 of Kallenberg \cite{K02} or
Proposition 7.1 of Landriault et al. \cite{LLZ16}, we know that Assumption
(\textbf{A2}) implies that the measures $\frac{1}{\varepsilon}B_{2}%
^{(q)}(x,\mathrm{d}z;x-a,x+\varepsilon)$, $\frac{1}{\varepsilon}B_{2}%
^{(q)}(x,\mathrm{d}z;x+\varepsilon-a,x+\varepsilon)$, $\frac{1}{\varepsilon
}B_{2}^{(q)}(x-\varepsilon,\mathrm{d}z;x-a,x)$ and $\frac{1}{\varepsilon}%
B_{2}^{(q)}(x-\varepsilon,\mathrm{d}z;x-\varepsilon-a,x)$ weakly converge to
the same measure on $%
\mathbb{R}
_{+}$, denoted as $b_{a,2}^{(q)}(x,\mathrm{d}z)$, such that $\int_{%
\mathbb{R}
_{+}}e^{-sz}b_{a,2}^{(q)}(x,\mathrm{d}z)=b_{a,2}^{(q,s)}(x)$. We point out
that it is possible that $b_{a,2}^{(q)}(x,\{0\})>0$, though the measure
$B_{2}^{(q)}(x,\mathrm{d}z;u,v)$ is only defined on $z\in(0,\infty)$.
\end{remark}

We are now ready to present the main result of this paper related to the joint
law of $(\tau_{a},Y_{\tau_{a}},M_{\tau_{a}})$.

\begin{theorem}
\label{thm markov}Consider a general time-homogeneous Markov process $X$
satisfying Assumptions (\textbf{A1})-(\textbf{A3}). For $q,s\geq0$ and
$K\in\mathbb{R}$, let
\[
h(x)=\mathbb{E}_{x}\left[  e^{-q\tau_{a}-s(Y_{\tau_{a}}-a)}1_{\{\tau_a<\infty, M_{\tau_{a}%
}\leq K\}}\right]  ,\quad x\leq K.
\]
Then $h(\cdot)$ is differentiable in $x<K$ and solves the following integral
equation
\begin{equation}
h(x)=\int_{x}^{K}e^{-\int_{x}^{y}b_{a,1}^{(q)}(w)\mathrm{d}w}\left(
c_{a}^{(q,s)}(y)+\int_{[0,K-y)}h(y+z)b_{a,2}^{(q)}(y,\mathrm{d}z)\right)
\mathrm{d}y\text{,}\quad x\leq K. \label{triple LT}%
\end{equation}

\end{theorem}

\begin{proof}
By the strong Markov property of $X$, for any $X_{0}=x\leq y<K$ and
$0<\varepsilon<(K-y)\wedge a$, we have
\begin{align*}
h(y)  &  =\mathbb{E}_{y}\left[  e^{-q\tau_{a}-s(Y_{\tau_{a}}-a)}1_{\left\{\tau_{a}<\infty, 
\tau_{a}<T_{y+\varepsilon}^{+}\right\}  }\right]  +\mathbb{E}_{y}\left[
e^{-qT_{y+\varepsilon}^{+}}1_{\{T_{y+\varepsilon}^{+}<\infty, T_{y+\varepsilon}^{+}<\tau_{a}%
,X_{T_{y+\varepsilon}^{+}}=y+\varepsilon\}}\right]  h(y+\varepsilon)\\
&  +\int_{0}^{K-y-\varepsilon}\mathbb{E}_{y}\left[  e^{-qT_{y+\varepsilon}%
^{+}}1_{\{T_{y+\varepsilon}^{+}<\infty, T_{y+\varepsilon}^{+}<\tau_{a},X_{T_{y+\varepsilon}^{+}%
}-y-\varepsilon\in\mathrm{d}z\}}\right]  h(y+\varepsilon+z).
\end{align*}
By Corollary \ref{cor bd}, it follows that%
\begin{align}
h(y+\varepsilon)-h(y)  &  \geq-e^{s\varepsilon}C^{(q,s)}(y;y+\varepsilon
-a,y+\varepsilon)+\left(  1-B_{1}^{(q)}(y;y-a,y+\varepsilon)\right)
h(y+\varepsilon)\nonumber\\
&  -\int_{0}^{K-y-\varepsilon}h(y+\varepsilon+z)B_{2}^{(q)}(y,\mathrm{d}%
z;y-a,y+\varepsilon), \label{down}%
\end{align}
and%
\begin{align}
h(y+\varepsilon)-h(y)  &  \leq-e^{-s\varepsilon}C^{(q,s)}(y;y-a,y+\varepsilon
)+\left(  1-B_{1}^{(q)}(y;y+\varepsilon-a,y+\varepsilon)\right)
h(y+\varepsilon)\nonumber\\
&  -\int_{0}^{K-y-\varepsilon}h(y+\varepsilon+z)B_{2}^{(q)}(y,\mathrm{d}%
z;y+\varepsilon-a,y+\varepsilon). \label{up}%
\end{align}
By Assumptions (\textbf{A1})-(\textbf{A3}) and $h(\cdot)\in\lbrack0,1]$, it is
clear that both the lower bound of $h(y+\varepsilon)-h(y)$ in (\ref{down}) and
the upper bound in (\ref{up}) vanish as $\varepsilon\downarrow0$. Hence,
$h(y)$ is right continuous for $y\in\lbrack x,K)$. Replacing $y$ by
$y-\varepsilon$ in (\ref{down}) and (\ref{up}), and using Assumptions
(\textbf{A1})-(\textbf{A3}) again, it follows that $h(y)$ is also left
continuous for $y\in(x,K]$ with $h(K)=0$. Therefore, $h(y)$ is continuous for
$y\in\lbrack x,K]$ (left continuous at $x$ and right continuous at $K$).

To consecutively show the differentiability, we divide inequalities
(\ref{down}) and (\ref{up}) by $\varepsilon$. It follows from Assumptions
(\textbf{A1})-(\textbf{A3}), Remark \ref{rk weak}, Lemma \ref{lem fatou} and
the continuity of $h$ that
\begin{align*}
&  \liminf_{\varepsilon\downarrow0}\frac{h(y+\varepsilon)-h(y)}{\varepsilon}\\
&  \geq-c_{a}^{(q,s)}(y)+b_{a,1}^{(q)}(y)h(y)-\limsup_{\varepsilon\downarrow
0}\int_{0}^{K-y-\varepsilon}h(y+\varepsilon+z)\frac{B_{2}^{(q)}(y,\mathrm{d}%
z;y-a,y+\varepsilon)}{\varepsilon}\\
&  \geq-c_{a}^{(q,s)}(y)+b_{a,1}^{(q)}(y)h(y)-\int_{[0,K-y)}h(y+z)b_{a,2}%
^{(q)}(y,\mathrm{d}z)\text{,}%
\end{align*}
and similarly,
\[
\limsup_{\varepsilon\downarrow0}\frac{h(y+\varepsilon)-h(y)}{\varepsilon}%
\leq-c_{a}^{(q,s)}(y)+b_{a,1}^{(q)}(y)h(y)-\int_{[0,K-y)}h(y+z)b_{a,2}%
^{(q)}(y,\mathrm{d}z).
\]
Since the two limits coincide, one concludes that $h(y)$ is right
differentiable for $y\in(x,K)$. Moreover, by replacing $y$ by $y-\varepsilon$
in (\ref{down}) and (\ref{up}), and using similar arguments, we can show that
$h(y)$ is also left differentiable for $y\in(x,K)$. Since the left and right
derivatives coincide, we conclude that $h(y)$ is differentiable for any
$y\in(x,K)$ and solves the following ordinary integro-differential equation
(OIDE),%
\begin{equation}
h^{\prime}(y)-b_{a,1}^{(q)}(y)h(y)=-c_{a}^{(q,s)}(y)-\int_{[0,K-y)}%
h(y+z)b_{a,2}^{(q)}(y,\mathrm{d}z). \label{h'}%
\end{equation}

Multiplying both sides of (\ref{h'}) by $e^{-\int_{x}^{y}b_{a,1}%
^{(q)}(w)\mathrm{d}w}$, integrating the resulting equation (with respect to
$y$) from $x$ to $K$, and using $h(K)=0$, this completes the proof of Theorem
\ref{thm markov}.\bigskip
\end{proof}

When the Markov process $X$ is spectrally negative (i.e., with no upward
jumps), the upward overshooting density $b_{a,2}^{(q)}(x,\mathrm{d}z)$ is
trivially $0$. Theorem \ref{thm markov} reduces to the following corollary.

\begin{corollary}
\label{cor snm}Consider a spectrally negative time-homogeneous Markov process
$X$ satisfying Assumptions (\textbf{A1}) and (\textbf{A3}). For $q,s\geq0$ and
$K>0$, we have%
\[
\mathbb{E}_{x}\left[  e^{-q\tau_{a}-s(Y_{\tau_{a}}-a)}1_{\{\tau_{a}<\infty, M_{\tau_{a}}\leq
K\}}\right]  =\int_{x}^{K}e^{-\int_{x}^{y}b_{a,1}^{(q)}(w)\mathrm{d}w}%
c_{a}^{(q,s)}(y)\mathrm{d}y\text{,}\quad x\leq K.
\]

\end{corollary}

When $X$ is a general L\'{e}vy process (possibly with two-sided jumps), we
have the following result for the joint Laplace transform of the triplet
$(\tau_{a},Y_{\tau_{a}},M_{\tau_{a}})$. Note that Corollary \ref{cor levy}
should be compared to Theorem 4.1 of Baurdoux \cite{B09}, in which, under the
L\'{e}vy framework, the resolvent density of $Y$ is expressed in terms of the
resolvent density of $X$ using excursion theory.

\begin{corollary}
\label{cor levy}Consider a L\'{e}vy process $X$ satisfying Assumptions
(\textbf{A1})-(\textbf{A3}). For $q,s,\delta\geq0$, we have\footnote{For L\'evy processes $\pr\{\tau_a<\infty\}=1$ as long as $X$ is not monotone.}
\begin{equation}
\mathbb{E}\left[  e^{-q\tau_{a}-s(Y_{\tau_{a}}-a)-\delta M_{\tau_{a}}}\right]
=\frac{c_{a}^{(q,s)}(0)}{\delta+b_{a}^{(q,\delta)}(0)}. \label{levy}%
\end{equation}

\end{corollary}

\begin{proof}
By the spatial homogeneity of the L\'{e}vy process $X$, Eq. (\ref{triple LT})
at $x=0$ reduces to
\begin{equation}
h(0)=\frac{c_{a}^{(q,s)}(0)}{b_{a,1}^{(q)}(0)}\left(  1-e^{-b_{a,1}^{(q)}%
(0)K}\right)  +\int_{0}^{K}e^{-b_{a,1}^{(q)}(0)y}\int_{[0,K-y)}h(y+z)b_{a,2}%
^{(q)}(0,\mathrm{d}z)\mathrm{d}y. \label{h}%
\end{equation}
Let
\[
\hat{h}(0):=\mathbb{E}\left[  e^{-q\tau_{a}-s(Y_{\tau_{a}}-a)-\delta
M_{\tau_{a}}}\right]  =\mathbb{E}\left[  e^{-q\tau_{a}-s(Y_{\tau_{a}}%
-a)}1_{\{M_{\tau_{a}}\leq e_{\delta}\}}\right]  \text{,}%
\]
where $e_{\delta}$ is an independent exponential random variable with finite
mean $1/\delta>0$. Multiplying both sides of (\ref{h}) by $\delta e^{-\delta
K}$, integrating the resulting equation (with respect to $K$) from $0$ to
$\infty$, and using integration by parts, one obtains
\begin{align*}
\hat{h}(0)  &  =\frac{c_{a}^{(q,s)}(0)}{\delta+b_{a,1}^{(q)}(0)}+\int
_{0}^{\infty}\delta e^{-\delta K}\int_{0}^{K}e^{-b_{a,1}^{(q)}(0)y}%
\int_{[0,K-y)}h(y+z)b_{a,2}^{(q)}(0,\mathrm{d}z)\mathrm{d}y\mathrm{d}K\\
&  =\frac{c_{a}^{(q,s)}(0)}{\delta+b_{a,1}^{(q)}(0)}+\int_{0}^{\infty
}e^{-b_{a,1}^{(q)}(0)y}\mathrm{d}y\int_{%
\mathbb{R}
_{+}}b_{a,2}^{(q)}(0,\mathrm{d}z)\int_{z+y}^{\infty}\delta e^{-\delta
K}\mathbb{E}\left[  e^{-q\tau_{a}-s(Y_{\tau_{a}}-a)}1_{\{M_{\tau_{a}}\leq
K-y-z\}}\right]  \mathrm{d}K\\
&  =\frac{c_{a}^{(q,s)}(0)}{\delta+b_{a,1}^{(q)}(0)}+\hat{h}(0)\frac{\int_{%
\mathbb{R}
_{+}}e^{-\delta z}b_{a,2}^{(q)}(0,\mathrm{d}z)}{\delta+b_{a,1}^{(q)}(0)}.
\end{align*}
Solving for $\hat{h}(0)$ and using (\ref{bbb}), it follows that%
\[
\hat{h}(0)=\frac{c_{a}^{(q,s)}(0)}{\delta+b_{a,1}^{(q)}(0)-\int_{%
\mathbb{R}
_{+}}e^{-\delta z}b_{a,2}^{(q)}(0,\mathrm{d}z)}=\frac{c_{a}^{(q,s)}(0)}%
{\delta+b_{a}^{(q,\delta)}(0)}.
\]
It follows from the monotone convergence theorem that (\ref{levy}) also holds
for $\delta=0$.\bigskip
\end{proof}

\begin{remark}
\label{rk levy} \normalfont
We point out that Assumptions (\textbf{A1})-(\textbf{A3}) are not necessary to
yield (\ref{levy}) in the L\'{e}vy framework. In fact, by the spatial
homogeneity of $X$, similar to (\ref{down}) and (\ref{up}), we have
\[
\frac{e^{-(s+\delta)\varepsilon}C^{(q,s)}(0;-a,\varepsilon)}{1-e^{-\delta
\varepsilon}B^{(q,\delta)}(0;\varepsilon-a,\varepsilon)}\leq\mathbb{E}\left[
e^{-q\tau_{a}-s(Y_{\tau_{a}}-a)-\delta M_{\tau_{a}}}\right]  \leq
\frac{e^{s\varepsilon}C^{(q,s)}(0;\varepsilon-a,\varepsilon)}{1-e^{-\delta
\varepsilon}B^{(q,\delta)}(0;-a,\varepsilon)},
\]
for any $\varepsilon\in(0,a)$. Suppose that the following condition holds:
\[
\lim_{\varepsilon\downarrow0}\frac{C^{(q,s)}(0;-a,\varepsilon)}{1-e^{-\delta
\varepsilon}B^{(q,\delta)}(0;\varepsilon-a,\varepsilon)}=\lim_{\varepsilon
\downarrow0}\frac{C^{(q,s)}(0;\varepsilon-a,\varepsilon)}{1-e^{-\delta
\varepsilon}B^{(q,\delta)}(0;-a,\varepsilon)}:=D_{a}^{(q,s,\delta)}%
\]
Then,
\[
\mathbb{E}\left[  e^{-q\tau_{a}-s(Y_{\tau_{a}}-a)-\delta M_{\tau_{a}}}\right]
=D_{a}^{(q,s,\delta)}.
\]

\end{remark}

Theorem \ref{thm markov} shows that the joint law $\mathbb{E}_{x}\left[
e^{-q\tau_{a}-s(Y_{\tau_{a}}-a)}1_{\{M_{\tau_{a}}\leq K\}}\right]  $ is a
solution to Eq. (\ref{triple LT}). Furthermore, the following theorem shows
that Eq. (\ref{triple LT}) admits a unique solution.

\begin{theorem}
Suppose that Assumptions (\textbf{A1})-(\textbf{A3}) hold. For $q,s\geq0$ and
$K>0$, Eq. (\ref{triple LT}) admits a unique solution.
\end{theorem}

\begin{proof}
From Theorem \ref{thm markov}, we know that $h(x):=\mathbb{E}_{x}\left[
e^{-q\tau_{a}-s(Y_{\tau_{a}}-a)}1_{\{\tau_{a}<\infty, M_{\tau_{a}}\leq K\}}\right]  $ is a
solution of (\ref{triple LT}). We also notice that any continuous solution to
(\ref{triple LT}) must vanish when $x\uparrow$ $K$. For any fixed
$L\in(-\infty,K)$, we define a metric space $(\mathbb{A}_{L},\boldsymbol{d}%
_{L})$, where $\mathbb{A}_{L}=\left\{  f\in C[L,K],f(K)=0\right\}  $ and the
metric $\boldsymbol{d}_{L}(f,g)=\sup_{x\in\lbrack L,K]}|f(x)-g(x)|$ for
$f,g\in\mathbb{A}_{L}$. We then define a mapping $\mathcal{L}$ on
$\mathbb{A}_{L}$ by
\[
\mathcal{L}f(x)=\int_{x}^{K}e^{-\int_{x}^{y}b_{a,1}^{(q)}(w)\mathrm{d}%
w}\left(  c_{a}^{(q,s)}(y)+\int_{[0,K-y)}f(y+z)b_{a,2}^{(q)}(y,\mathrm{d}%
z)\right)  \mathrm{d}y,\text{\quad}x\in\lbrack L,K],
\]
where $f\in\mathbb{A}_{L}$. It is clear that $\mathcal{L}(\mathbb{A}%
_{L})\subset\mathbb{A}_{L}$.

Next we show that $\mathcal{L}:\mathbb{A}_{L}\rightarrow\mathbb{A}_{L}$ is a
contraction mapping. By the definitions of the two-sided exit quantities, for
any $y\in%
\mathbb{R}
$, it follows that%
\begin{equation}
C^{(q,s)}(y;y-a,y+\varepsilon)+\int_{%
\mathbb{R}
_{+}}B_{2}^{(q)}(y,\mathrm{d}z;y-a,y+\varepsilon)\leq1-B_{1}^{(q)}%
(y;y-a,y+\varepsilon). \label{C1B}%
\end{equation}
Dividing each term in (\ref{C1B}) by $\varepsilon\in(0,a)$ and letting
$\varepsilon\downarrow0$, it follows from Assumptions (\textbf{A1}%
)-(\textbf{A3}) that%
\begin{equation}
0\leq c_{a}^{(q,s)}(y)+\int_{%
\mathbb{R}
_{+}}b_{a,2}^{(q)}(y,\mathrm{d}z)\leq b_{a,1}^{(q)}(y),\quad y\in%
\mathbb{R}
. \label{ineq}%
\end{equation}
By (\ref{ineq}), we have for any $f,g\in\mathbb{A}_{L}$,
\begin{align*}
\boldsymbol{d}_{L}\left(  \mathcal{L}f,\mathcal{L}g\right)   &  \leq\sup
_{t\in\lbrack L,K]}\left\vert f(t)-g(t)\right\vert \sup_{x\in\lbrack L,K]}%
\int_{x}^{K}e^{-\int_{x}^{y}b_{a,1}^{(q)}(w)\mathrm{d}w}\int_{%
\mathbb{R}
_{+}}b_{a,2}^{(q)}(y,\mathrm{d}z)\mathrm{d}y\\
&  \leq\boldsymbol{d}_{L}(f,g)\sup_{L\leq x\leq K}\int_{x}^{K}e^{-\int_{x}%
^{y}b_{a,1}^{(q)}(w)\mathrm{d}w}b_{a,1}^{(q)}(y)\mathrm{d}y\\
&  \leq\boldsymbol{d}_{L}(f,g)\left(  1-e^{-\int_{L}^{K}b_{a,1}^{(q)}%
(w)\mathrm{d}w}\right)  .
\end{align*}
Since $\int_{L}^{K}b_{a,1}^{(q)}(w)\mathrm{d}w<\infty$ by Assumption
(\textbf{A1}), one concludes that $\mathcal{L}:\mathbb{A}_{L}\rightarrow
\mathbb{A}_{L}$ is a contraction mapping. By Banach fixed point theorem, there
exists a unique fixed point in $\mathbb{A}_{L}$. By a restriction of domain,
it is easy to see that $\mathbb{A}_{L_{1}}\subset\mathbb{A}_{L_{2}}$ for
$-\infty<L_{1}<L_{2}<K$. By the arbitrariness of $L$, the uniqueness holds for
the space $\cap_{L<K}\mathbb{A}_{L}$. This completes the proof.
\end{proof}

For the reminder of this section, we state several examples of Markov
processes satisfying Assumptions (\textbf{A1})-(\textbf{A3}). Note that the
joint law of drawdown estimates for Examples \ref{eg SNLP} and
\ref{eg diffusion} were solved by Mijatovic and Pistorius \cite{MP12} and
Lehoczky \cite{L77}, respectively (using different approaches). Assumption
verifications for Examples \ref{eg PEMP} and \ref{eg JD} are postponed to Appendix.

\begin{example}
[Spectrally negative L\'{e}vy processes]\label{eg SNLP} \normalfont Consider a
spectrally negative L\'{e}vy process $X$. Let $\psi(s):=\frac{1}{t}%
\log\mathbb{E}[e^{sX_{t}}]$ $\left(  s\geq0\right)  $ be the Laplace exponent
of $X$. Further, let $W^{(q)}:%
\mathbb{R}
\rightarrow\lbrack0,\infty)$ be the well-known $q$-scale function of $X$; see,
for instance Chapter 8 of Kyprianou \cite{K14}. The second scale function is
defined as $Z^{(q)}(x)=1+q\int_{0}^{x}W^{(q)}(y)\mathrm{d}y$. Under some mild
conditions (e.g., Lemma 2.4 of Kuznetsov et al. \cite{KKR12}), the scale
functions are continuously differentiable which further implies that
Assumptions (\textbf{A1}) and (\textbf{A3}) hold with
\begin{equation}
b_{a,1}^{(q)}(0)=\frac{W^{(q)\prime}(a)}{W^{(q)}(a)}\text{ and }c_{a}%
^{(q,s)}(0)=e^{sa}\frac{Z_{s}^{(p)}(a)W_{s}^{(p)\prime}(a)-Z_{s}^{(p)\prime
}(a)W_{s}^{(p)}(a)}{W_{s}^{(p)}(a)}, \label{bc}%
\end{equation}
where $p=q-\psi(s)$, and $W_{s}^{(p)}$ ($Z_{s}^{(p)}$) is the (second) scale
function of $X$ under a new probability measure\ $\mathbb{P}^{s}$ defined by
the Radon-Nikodym derivative process $\left.  \frac{\mathrm{d}\mathbb{P}^{s}%
}{\mathrm{d}\mathbb{P}}\right\vert _{\mathcal{F}_{t}}=e^{sX_{t}-\psi(s)t}$ for
$t\geq0$. Therefore, by Corollary \ref{cor levy} and (\ref{bc}), we have%
\[
\mathbb{E}\left[  e^{-q\tau_{a}-s(Y_{\tau_{a}}-a)-\delta M_{\tau_{a}}}\right]
=\frac{e^{sa}W^{(q)}(a)}{\delta W^{(q)}(a)+W^{(q)\prime}(a)}\frac{Z_{s}%
^{(p)}(a)W_{s}^{(p)\prime}(a)-pW_{s}^{(p)}(a)^{2}}{W_{s}^{(p)}(a)},
\]
which is consistent with Theorem 3.1 of Landriault et al. \cite{LLZ16}, and
\textcolor[rgb]{0.0,0.0,0.0}{Theorem 1 of
Mijatovic and Pistorius \cite{MP12}}.
\end{example}

\begin{example}
[Refracted L\'{e}vy processes]\label{eg refracted} \normalfont Consider a
refracted spectrally negative L\'{e}vy process $X$ of the form
\begin{equation}
X_{t}=U_{t}-\lambda\int_{0}^{t}1_{\{X_{s}>b\}}\mathrm{d}s, \label{RLEVY}%
\end{equation}
where $\lambda\geq0$, $b>0$, and $U$ is a spectrally negative L\'{e}vy process
(see Kyprianou and Loeffen \cite{KL10}). Let $W^{(q)}$ ($Z^{(q)}$) be the
(second) $q$-scale function of $U$, and $\mathbb{W}^{(q)}$ be the $q$-scale
function of the process $\{U_{t}-\lambda t\}_{t\geq0}$. Similar to Example
\ref{eg SNLP}, all the scale functions are continuously differentiable under
mild conditions.

For simplicity, we only consider the quantity
\textcolor[rgb]{0.0,0.0,0.0}{$\mathbb{E}_{x}\left[
e^{-q\tau_{a}}1_{\{\tau_a<\infty, M_{\tau_a}\le K\}}\right]  $} with $b>x-a$ (otherwise the
problem reduces to Example \ref{eg SNLP} for $X_{t}=U_{t}-\lambda t$). By
Theorem 4 of Kyprianou and Loeffen \cite{KL10}, one can verify that
Assumptions (\textbf{A1}) and (\textbf{A3}) hold. For $b>x$, from (\ref{bc})
with $s=0$, we have
\[
b_{a,1}^{(q)}(x)=\frac{W^{(q)\prime}(a)}{W^{(q)}(a)}\text{ and }c_{a}%
^{(q,0)}(x)=\frac{Z^{(q)}(a)W^{(q)\prime}(a)-Z^{(q)\prime}(a)W^{(q)}%
(a)}{W^{(q)}(a)}.
\]
For $x>b>x-a$,%
\[
b_{a,1}^{(q)}(x)=\frac{\left(  1+\lambda\mathbb{W}^{(q)}(0)\right)
W^{(q)\prime}(a)+\lambda\int_{b-x+a}^{a}\mathbb{W}^{(q)\prime}%
(a-y)W^{(q)\prime}(y)\mathrm{d}y}{W^{(q)}(a)+\lambda\int_{b-x+a}^{a}%
\mathbb{W}^{(q)}(a-y)W^{(q)\prime}(y)\mathrm{d}y}%
\]
and
\[
c_{a}^{(q,0)}(x)=\frac{k_{a}^{(q)}(x)}{W^{(q)}(a)+\lambda\int_{b-x+a}%
^{a}\mathbb{W}^{(q)}(a-y)W^{(q)\prime}(y)dy},
\]
where%
\begin{align*}
k_{a}^{(q)}(x)  &  =(1+\lambda\mathbb{W}^{(q)}(0))\left(  Z^{(q)}%
(a)W^{(q)\prime}(a)-qW^{(q)}(a)^{2}\right) \\
&  +\lambda q(1+\lambda\mathbb{W}^{(q)}(0))\int_{b-x+a}^{a}\mathbb{W}%
^{(q)}(a-y)\left(  W^{(q)\prime}(a)W^{(q)}(y)-W^{(q)}(a)W^{(q)\prime
}(y)\right)  \mathrm{d}y\\
&  -\lambda q\left[  W^{(q)}(a)+\lambda\int_{b-x+a}^{a}\mathbb{W}%
^{(q)}(a-y)W^{(q)\prime}(y)\mathrm{d}y\right]  \int_{b-x+a}^{a}\mathbb{W}%
^{(q)\prime}(a-y)W^{(q)}(y)\mathrm{d}y\\
&  +\lambda\left[  Z^{(q)}(a)+\lambda q\int_{b-x+a}^{a}\mathbb{W}%
^{(q)}(a-y)W^{(q)}(y)\mathrm{d}y\right]  \int_{b-x+a}^{a}\mathbb{W}%
^{(q)\prime}(a-y)W^{(q)\prime}(y)\mathrm{d}y.
\end{align*}
By Corollary \ref{cor snm}, we obtain
\[
\mathbb{E}_{x}\left[  e^{-q\tau_{a}}1_{\{M_{\tau_{a}}\leq K\}}\right]
=\int_{x}^{K}e^{-\int_{x}^{y}b_{a,1}^{(q)}(w)\mathrm{d}w}c_{a}^{(q,0)}%
(y)\mathrm{d}y\text{,}\quad x\leq K,
\]
which is a new result for the refracted L\'{e}vy process (\ref{RLEVY}).
\end{example}

\begin{example}
[Linear diffusion processes]\label{eg diffusion} \normalfont Consider a linear
diffusion process $X$ of the form
\[
\mathrm{d}X_{t}=\mu(X_{t})\mathrm{d}t+\sigma(X_{t})\mathrm{d}W_{t},
\]
where $(W_{t})_{t\geq0}$ is a standard Brownian motion, and the drift term
$\mu(\cdot)$ and local volatility $\sigma(\cdot)>0$ satisfy the usual
Lipschitz continuity and linear growth conditions. As a special case of the
jump diffusion process of Example \ref{eg JD}, it will be shown later that
Assumptions (\textbf{A1}) and (\textbf{A3}) hold for linear diffusion
processes. By Corollary \ref{cor snm}, we obtain
\[
\mathbb{E}_{x}\left[  e^{-q\tau_{a}}1_{\{\tau_a<\infty, M_{\tau_{a}}\leq K\}}\right]
=\int_{x}^{K}e^{-\int_{x}^{y}b_{a,1}^{(q)}(w)\mathrm{d}w}c_{a}^{(q,0)}%
(y)\mathrm{d}y\text{,}\quad x\leq K,
\]
which is consistent with Eq. (4) of Lehoczky \cite{L77}.
\end{example}

\begin{example}
[Piecewise exponential Markov processes]\label{eg PEMP} \normalfont Consider a
piecewise exponential Markov process (PEMP) $X$ of the form
\begin{equation}
\mathrm{d}X_{t}=\mu X_{t}\mathrm{d}t+\mathrm{d}Z_{t}, \label{PEMP}%
\end{equation}
where $\mu>0$ is the drift coefficient and $Z=(Z_{t})_{t\geq0}$ is a compound
Poisson process given by $Z_{t}=\sum_{i=1}^{N_{t}}J_{i}$. Here, $(N_{t}%
)_{t\geq0}$ is a Poisson process with intensity $\lambda>0$ and $J_{i}$'s are
iid copies of a real-valued random variable $J$ with cumulative distribution
function $F$.
\BIN{We also assume the initial value $X_0\geq a$ which ensures that $X_t\geq 0$ for all $t<\tau_a$. In this case, as discussed in Remark \ref{rmk31}, $X$ is upward regular and creeps upward before $\tau_a$.}
The first passage times of $X$ have been extensively studied in applied
probability; see, e.g., Tsurui and Osaki \cite{TO76} and Kella and Stadje
\cite{KS01}. For the PEMP (\ref{PEMP}), semi-explicit expressions for the
two-sided exit quantities $B_{1}^{(q)}(\cdot)$, $B_{2}^{(q)}(\cdot,\cdot)$ and
$C^{(q,s)}(\cdot)$ are given in Section 6 of Jacobsen and Jensen \cite{JJ07}.
As will be shown in Section \ref{ver pemp}, Assumptions (\textbf{A1}%
)-(\textbf{A3}) and Theorem \ref{thm markov} hold for the PEMP $X$ with a
continuous jump size distribution $F$.
\end{example}

\begin{example}
[Jump diffusion]\label{eg JD} \normalfont Consider a jump diffusion process
$X$ of the form
\begin{equation}
\mathrm{d}X_{t}=\mu(X_{t})\mathrm{d}t+\sigma(X_{t})\mathrm{d}W_{t}%
+\int_{-\infty}^{\infty}\gamma(X_{t-},z)N(\mathrm{d}t,\mathrm{d}z), \label{JD}%
\end{equation}
where $\mu(\cdot)$ and $\sigma(\cdot)>0$ are functions on $\mathbb{R}$,
$(W_{t})_{t\geq0}$ is a standard Brownian motion, $\gamma(\cdot,\cdot)$ is a
real-valued function on $\mathbb{R}^{2}$ modeling the jump size, and
$N(\mathrm{d}t,\mathrm{d}z)$ is an independent Poisson random measure on
$\mathbb{R}_{+}\times\mathbb{R}$ with a finite intensity measure
\textcolor[rgb]{0.0,0.0,0.0}{$\mathrm{d}t\times\nu(\mathrm{d}z)$}. For
specific $\mu(\cdot)$ and $\sigma(\cdot)$, the jump diffusion (\ref{JD}) can
be used to model the surplus process of an insurer with investment in risky
assets; see, e.g., Gjessing and Paulsen \cite{GP97} and Yuen et al.
\cite{YWN04}. We assume the same conditions as Theorem 1.19 of \O ksendal and
Sulem-Bialobroda \cite{OS07} so that (\ref{JD}) admits a unique c\`{a}dl\`{a}g
adapted solution. Under this setup, we show in Section \ref{ver JD} that
Assumptions (\textbf{A1})--(\textbf{A3}) and thus Theorem \ref{thm markov}
hold for the jump diffusion (\ref{JD}).
\end{example}

\section{Numerical examples}

\BIN{The main results of Section 3 rely on the analytic tractability of the
two-sided exit quantities. To further illustrate their applicability, we now
consider the numerical evaluation of the joint law of $(Y_{\tau_{a}}%
,M_{\tau_{a}})$ for two particular spatial-inhomogeneous Markov processes with
(positive) jumps through Theorem \ref{thm markov}. For simplicity, we assume
that the discount rate $q=0$ throughout this section.}

\subsection{PEMP}

\BIN{In this section, we consider the PEMP $X$ in Example \ref{eg PEMP} with
$\mu=1$, $\lambda=3$, and the generic jump size $J$ with density
\begin{equation}
p(x)=\left\{
\begin{array}
[c]{ll}%
\frac{1}{3}e^{-x}, & x>0,\\
\frac{1}{3}(e^{x}+2e^{2x}), & x<0.
\end{array}
\right.  \label{p1}%
\end{equation}
We follow Section 6 of Jacobsen and Jensen \cite{JJ07} to first solve for the
two-sided exit quantities. Define the integral kernel
\[
\psi_{0}(z):=\frac{1}{z(z+1)(z-1)(z-2)},\quad z\in%
\mathbb{C}
,
\]
and the linearly independent functions%
\[%
\begin{array}
[c]{ll}%
g_{1}(x):=\frac{1}{2\pi\sqrt{-1}}\int_{\Gamma_{1}}\psi_{0}(z)e^{-xz}%
dz=\frac{1}{6}e^{-2x}, & g_{2}(x):=\frac{1}{2\pi\sqrt{-1}}\int_{\Gamma_{2}%
}\psi_{0}(z)e^{-xz}dz=-\frac{1}{2}e^{-x},\\
g_{3}(x):=\frac{1}{2\pi\sqrt{-1}}\int_{\Gamma_{3}}\psi_{0}(z)e^{-xz}%
dz=\frac{1}{2}, & g_{4}(x):=\frac{1}{2\pi\sqrt{-1}}\int_{\Gamma_{4}}\psi
_{0}(z)e^{-xz}dz=-\frac{1}{6}e^{x},
\end{array}
\]
for $x>0,$ where $\Gamma_{i}$ ($i=1,2,3,4$) is a small counterclockwise circle
centered at the pole $\mu_{i}=3-i$ of $\psi_{0}(z)$. Moreover, for $0<u<v$, we
consider the matrix-valued function
\[
(M_{i,k}({u,v}))_{1\leq i,k\leq4}:=%
\begin{pmatrix}
-\frac{1}{3}e^{-2u}(u+\frac{11}{6}) & \frac{e^{-2u}}{6} & \frac{e^{-2v}}{18} &
g_{1}(v)\\
e^{-u} & \frac{e^{-u}}{2}(u+\frac{1}{2}) & -\frac{e^{-v}}{4} & g_{2}(v)\\
-\frac{1}{2} & -\frac{1}{2} & \frac{1}{2} & g_{3}(v)\\
\frac{e^{u}}{9} & \frac{e^{u}}{12} & \frac{e^{v}}{6}(v-\frac{11}{6}) &
g_{4}(v)
\end{pmatrix}
,
\]
where the matrix $M$ entries are chosen according to
\[
\left\{
\begin{array}
[c]{l}%
M_{i,k}(u,v)=\frac{\mu_{k}}{2\pi\sqrt{-1}}\int_{\Gamma_{i}}\frac{\psi_{0}%
(z)}{z-\mu_{k}}e^{-uz}\mathrm{d}z,\quad1\leq i\leq4, k=1,2,\\
M_{i,3}(u,v)=\frac{|\mu_{4}|}{2\pi\sqrt{-1}}\int_{\Gamma_{i}}\frac{\psi
_{0}(z)}{z-\mu_{4}}e^{-vz}\mathrm{d}z,\quad1\leq i\leq4.
\end{array}
\right.
\]
Let $(N_{k,j}(u,v))_{1\leq k,j\leq4}$ be the inverse of $(M_{i,k}(u,v))_{1\leq
i,k\leq4}$.
Combining Eq. (46) and a generalized Eq. (48) of Jacobsen and Jensen
\cite{JJ07} (with $\zeta=s\geq0$ and $\rho\geq0$), we obtain the linear system
of equations
\begin{equation}
(c_{1},c_{2},c_{3},c_{4})(M_{i,k})=\left(  -\frac{2\underline{C}}{s+2}%
,-\frac{\underline{C}}{s+1},\frac{\overline{C}}{\rho+1},f(v)\right)  ,
\label{CCCC}%
\end{equation}
where $\underline{C}$ and\ $\overline{C}$ are constants specified later, and
$f(x)$ could stand for any of $B_{1}^{(0)}(x;u,v)$, $B_{2}^{(0,\rho)}(x;u,v)$,
or $C^{(0,s)}(x;u,v)$ and has the representation
\[
f(x)=\sum\limits_{i=1}^{4}c_{i}g_{i}(x),\quad x\in\lbrack u,v].
\]
}

\BIN{To solve for $B_{1}^{(0)}(x;u,v)$, $B_{2}^{(0,\rho)}(x;u,v)$, or
$C^{(0,s)}(x;u,v)$, we only need to solve (\ref{CCCC}) with different assigned
values of $\underline{C}$, $\overline{C}$, and $f(v)$ according to Eq. (45) of
Jacobsen and Jensen \cite{JJ07}. By letting $\underline{C}=\overline{C}=0$ and
$f(v)=1$, we obtain
\[
B_{1}^{(0)}(x;u,v)=\sum_{i=1}^{4}N_{4,i}(u,v)g_{i}(x).
\]
Similarly, by letting $\underline{C}=f(v)=0$ and $\overline{C}=1$, for
$\rho\geq0$, we obtain
\[
B_{2}^{(0,\rho)}(x;u,v)=\frac{1}{1+\rho}\sum_{i=1}^{4}N_{3,i}(u,v)g_{i}(x).
\]
A Laplace inversion with respect to $\rho$ yields, for $z>0$,
\[
B_{2}^{(0)}(x,\mathrm{d}z;u,v)=e^{-z}\sum_{i=1}^{4}N_{3,i}(u,v)g_{i}%
(x)\mathrm{d}z.
\]
By letting $\underline{C}=1$ and $\overline{C}=f(v)=0$, for $s\geq0$, we
obtain
\[
C^{(0,s)}(x;u,v)=\sum_{i=1}^{4}\left(  \frac{-2}{s+2}N_{1,i}(u,v)+\frac
{-1}{s+1}N_{2,i}(u,v)\right)  g_{i}(x).
\]
By the definitions, we have
\begin{align*}
b_{a,1}^{(0)}(x)  &  =-\sum_{i=1}^{4}D_{4,i}(x-a,x)g_{i}(x),\\
b_{a,2}^{(0)}(x,\mathrm{d}z)  &  =e^{-z}\left(  \sum_{i=1}^{4}D_{3,i}%
(x-a,x)g_{i}(x)\right)  \mathrm{d}z,\\
c_{a}^{(0,s)}(x)  &  =\sum_{i=1}^{4}\left(  \frac{-2}{s+2}D_{1,i}%
(x-a,x)+\frac{-1}{s+1}D_{2,i}(x-a,x)\right)  g_{i}(x),
\end{align*}
where we denote $D_{k,j}(u,v):=\frac{\partial}{\partial v}N_{k,j}(u,v)$.}

\BIN{In Figure \ref{fig1} below, we use \textsf{Mathematica} to numerically solve
the integral equation (\ref{triple LT}).}

\begin{figure}[H]
\centering
{{\includegraphics[width=220pt]{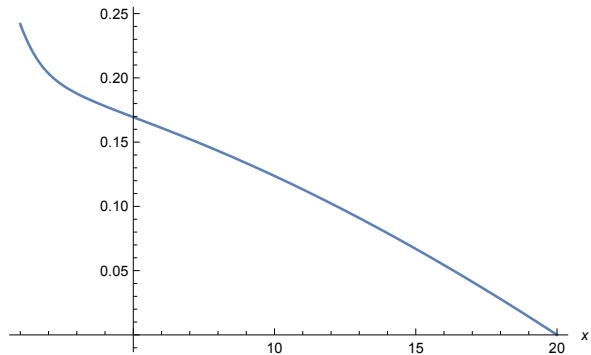} }}
\caption{Plot of the probability $h(x)=\mathbb{P}_{x}\{M_{\tau_{a}}\le K\}$ for PEMP \eqref{PEMP} with $q=0,\mu=1,\lambda
=3,a=1,K=20$ and jump size distribution given in \eqref{p1}}%
\label{fig1}%
\end{figure}

\subsection{A jump diffusion model}

\BIN{In this section, we consider a generalized PEMP $(X_{t})_{t\geq0}$ with
diffusion whose dynamics is governed by
\begin{equation}
\mathrm{d}X_{t}=X_{t}\mathrm{d}t+\sqrt{2}\mathrm{d}W_{t}+\mathrm{d}Z_{t},\quad
t>0, \label{eq:jd}%
\end{equation}
where the initial value $X_{0}=x\in%
\mathbb{R}
$, $(W_{t})_{t\geq0}$ is a standard Brownian motion, and $(Z_{t})_{t\geq0}$ is
an independent compound Poisson process with a unit jump intensity and a unit
mean exponential jump distribution. The two-sided exit quantities of this
generalized PEMP can also be solved using the approach described in Sections 6
and 7 of Jacobsen and Jensen \cite{JJ07}.}

\BIN{We define an integral kernel
\[
\psi_{1}(z)=\frac{e^{\frac{z^{2}}{2}}}{z(z+1)},\quad z\in%
\mathbb{C}
\text{.}%
\]
Let $\Gamma_{i}$ $\left(  i=1,2\right)  $ be small counterclockwise circles
around the simple poles $\mu_{1}=0$ and $\mu_{2}=-1$, respectively, and define
the linearly independent functions%
\begin{align*}
g_{1}(x)  &  :=\frac{1}{2\pi\sqrt{-1}}\int_{\Gamma_{1}}\psi_{1}(z)e^{-xz}%
dz=1,\\
g_{2}(x)  &  :=\frac{1}{2\pi\sqrt{-1}}\int_{\Gamma_{2}}\psi_{1}(z)e^{-xz}%
dz=-e^{x+\frac{1}{2}},
\end{align*}
for $x\in%
\mathbb{R}
$. To find another linearly independent partial eigenfunction, we consider the
vertical line $\Gamma_{3}=\{1+t\sqrt{-1},t\in%
\mathbb{R}
\}$ and define
\begin{equation}
g_{3}(x):=\frac{1}{2\pi\sqrt{-1}}\int_{\Gamma_{3}}\psi_{1}(z)e^{-xz}%
\mathrm{d}z. \label{eq:G3last}%
\end{equation}
Next we derive an explicit expression for $g_{3}(x)$. We know from
(\ref{eq:G3last}) that $\lim_{x\rightarrow\infty}g_{3}(x)=0$ and $g_{3}$ is
continuously differentiable with
\begin{equation}
g_{3}^{\prime}(x)=-\frac{1}{2\pi\sqrt{-1}}\int_{\Gamma_{3}}\frac
{e^{\frac{z^{2}}{2}}}{z+1}e^{-xz}\mathrm{d}z. \label{eq:last1}%
\end{equation}
Notice that the bilateral Laplace transform functions (e.g., Chapter VI of
\cite{W46}) of a standard normal random variable $U_{1}$ and an independent
unit mean exponential random variables $U_{2}$ are given respectively by
\[
\int_{-\infty}^{\infty}e^{-zy}\cdot\frac{1}{\sqrt{2\pi}}e^{-\frac{y^{2}}{2}%
}\mathrm{d}y=e^{\frac{z^{2}}{2}},\quad\int_{0}^{\infty}e^{-zy}\cdot
e^{-y}\mathrm{d}y=\frac{1}{z+1},
\]
for all complex $z$ such that $\Re(z)\geq0$. Hence, the bilateral Laplace
transform of the density function of $U_{1}+U_{2}$, i.e.,
\[
\int_{0}^{\infty}\frac{1}{\sqrt{2\pi}}e^{-\frac{(x-y)^{2}}{2}}e^{-y}%
\mathrm{d}y
\]
is given by $e^{\frac{z^{2}}{2}}/(z+1)$ for all complex $z$ such that
$\Re(z)\geq0$. Since the right hand side of (\ref{eq:last1}) is just the
Bromwich integral for the inversion of the bilateral Laplace transform
$-e^{\frac{z^{2}}{2}}/(z+1)$, evaluated at $-x$, we deduce that
\[
g_{3}^{\prime}(x)=-\int_{0}^{\infty}\frac{1}{\sqrt{2\pi}}e^{-\frac{(x+y)^{2}%
}{2}}e^{-y}\mathrm{d}y.
\]
It follows that
\[
g_{3}(x)=-\int_{x}^{\infty}g_{3}^{\prime}(y)\mathrm{d}y=1-\int_{0}^{\infty
}N(x+y)e^{-y}\mathrm{d}y.
\]
where $N(\cdot)$ is the cumulative distribution function of standard normal distribution.}

\BIN{For any fixed $-\infty<u<v<\infty$, we define a matrix-valued function
\[
(M_{i,k}(u,v))_{1\leq i,k\leq3}:=%
\begin{pmatrix}
1 & g_{1}(v) & g_{1}(u)\\
ve^{v+\frac{1}{2}} & g_{2}(v) & g_{2}(u)\\
1-\int_{0}^{\infty}N(v+y)ye^{-y}\mathrm{d}y & g_{3}(v) & g_{3}(u)
\end{pmatrix}
,
\]
where the first row is computed according to
\[
M_{i,1}(u,v)=\frac{1}{2\pi\sqrt{-1}}\int_{\Gamma_{i}}\frac{\psi_{0}(z)}%
{z+1}e^{-vz}\mathrm{d}z.
\]
Notice that $M_{3,1}(u,v)$ can be calculated in the same way as $g_{3}(x)$. We
also denote by $(N_{k,j}(u,v))_{1\leq k,j\leq3}$ the inverse of $(M_{i,k}%
(u,v))_{1\leq i,k\leq3}$.}

\BIN{By Eq. (46) and a generalized Eq. (48) of Jacobsen and Jensen \cite{JJ07}
(with $\zeta=s=0$ and $\rho\geq0$), we obtain the linear system of equations
\begin{equation}
(c_{1},c_{2},c_{3})(M_{i,k})=\left(  \frac{\overline{C}}{\rho+1}%
,f(v),f(u)\right)  , \label{CCC}%
\end{equation}
where $\overline{C}$ is a constant specified later, and $f(x)$ could stand for
any of $B_{1}^{(0)}(x;u,v)$, $B_{2}^{(0,\rho)}(x;u,v)$, or $C^{(0,0)}(x;u,v)$
and has the representation
\[
f(x)=\sum\limits_{i=1}^{3}c_{i}g_{i}(x),\quad x\in\lbrack u,v].
\]
By letting (1) $\overline{C}=f(u)=0$ and $f(v)=1$, (2) $\overline{C}=1$ and
$f(v)=f(u)=0$, (3) $\overline{C}=f(v)=0$ and $f(u)=1$, for any $\rho\geq0$ and
$z>0$, and solving the linear system (\ref{CCC}), we respectively obtain
\begin{align*}
B_{1}^{(0)}(x;u,v)  &  =\sum_{i=1}^{3}N_{2,i}(u,v)g_{i}(x),\\
B_{2}^{(0,\rho)}(x;u,v)  &  =\frac{1}{1+\rho}\sum_{i=1}^{3}N_{1,i}%
(u,v)g_{i}(x),\quad B_{2}^{(0)}(x,\mathrm{d}z;u,v)=e^{-z}\sum_{i=1}^{3}%
N_{1,i}(u,v)g_{i}(x)\mathrm{d}z,\\
C^{(0,0)}(x;u,v)  &  =\sum_{i=1}^{3}N_{3,i}(u,v)g_{i}(x).
\end{align*}
Furthermore, this implies
\begin{align*}
b_{a,1}^{(0)}(x)  &  =-\sum_{i=1}^{3}D_{2,1}(x-a,x)g_{i}(x),\\
b_{a,2}^{(0)}(x,\mathrm{d}z)  &  =e^{-z}\left(  \sum_{i=1}^{3}D_{1,i}%
(x-a,x)g_{i}(x)\right)  ,\\
c_{a}^{(0,0)}(x)  &  =\sum_{i=1}^{3}D_{3,i}(x-a,x)g_{i}(x),
\end{align*}
where we denote $D_{k,j}(u,v)=\frac{\partial}{\partial v}N_{k,j}(u,v)$.}

\BIN{In Figure 2 below, we plot $h(x)=\mathbb{P}_{x}\{M_{\tau_{a}}\leq K\}$ by
numerically solving the integral equation (\ref{triple LT}) using
\textsf{Mathematica}.}

\begin{figure}[H]
\centering
{\includegraphics[width=220pt]{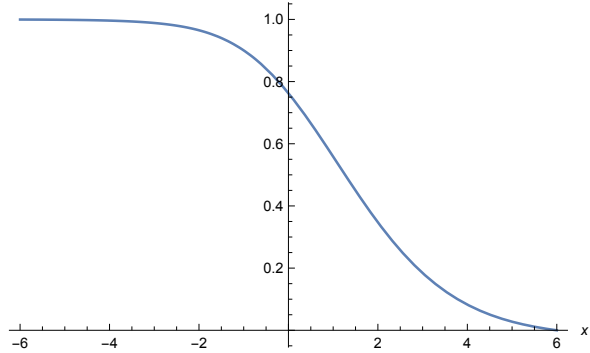}}\caption{Plot of the probability
$h(x)=\mathbb{P}_{x}\{M_{\tau_{a}}\le K\}$ for the jump diffusion in
\eqref{eq:jd} with $K=6$ and $a=1$.}%
\label{fig2}%
\end{figure}

\section{Acknowledgments}

The authors would like to thank two anonymous referees for their helpful
comments and suggestions. Support from grants from the Natural Sciences and
Engineering Research Council of Canada is gratefully acknowledged by David
Landriault and Bin Li (grant numbers 341316 and 05828, respectively). Support
from a start-up grant from the University of Waterloo is gratefully
acknowledged by Bin Li, as is support from the Canada Research Chair Program
by David Landriault.

\appendix

\section{Appendix}

\subsection{Proof of Lemma \ref{lem fatou}}

We define $\psi _{n}(z)=\inf_{m\geq n}\phi _{m}(z)$ for $z\in S$. Further,
we define $\underline{\psi }_{n}(z)=\liminf_{w\rightarrow z}\psi _{n}(w)$
which is lower semi-continuous (see, e.g., Lemma 5.13.4 of Berberian \cite%
{B99}). Note that $\underline{\psi }_{n}$ is increasing in $n$, and by the definition of $\underline{\psi}_n$, we have 
\begin{align*}
\lim_{n\rightarrow \infty }\underline{\psi }_{n}(z)=&\lim_{n\rightarrow
	\infty }\lim_{r\downarrow 0}\inf_{w\in (z-r,z+r)}\inf_{m\geq n}\phi
_{m}(w)\\
=&\lim_{n\rightarrow
	\infty }\lim_{r\downarrow 0}\inf_{m\geq n,w\in (z-r,z+r)}\phi
_{m}(w)\equiv\liminf_{n\rightarrow \infty ,w\rightarrow z}\phi _{n}(w),
\end{align*}%
where the second equality is because there is no ambiguity in switching the order of two infimums. 
By the monotone convergence theorem, we have 
\begin{equation}
\int_{S}\liminf_{n\rightarrow \infty ,w\rightarrow z}\phi _{n}(w)\mathrm{d}%
\mu (z)=\lim_{n\rightarrow \infty }\int_{S}\underline{\psi }_{n}(z)\mathrm{d}%
\mu (z).  \label{lim1}
\end{equation}%
By Portmanteau theorem of weak convergence and the fact that $\underline{%
	\psi }_{n}(z)$ is nonnegative and lower semi-continuous, it follows that 
\begin{equation}
\int_{S}\underline{\psi }_{n}(z)\mathrm{d}\mu (z)\leq \liminf_{m\rightarrow
	\infty }\int_{S}\underline{\psi }_{n}(z)\mathrm{d}\mu _{m}(z)  \label{lim2}
\end{equation}%
for any $n\in 
\mathbb{N}
$. Moreover, since $\psi _{n}(z)$ is monotone increasing in $n$, we have%
\begin{equation}
\liminf_{m\rightarrow \infty }\int_{S}\underline{\psi }_{n}(z)\mathrm{d}\mu
_{m}(z)\leq \liminf_{m\rightarrow \infty }\int_{S}\underline{\psi }_{m}(z)%
\mathrm{d}\mu _{m}(z).  \label{lim3}
\end{equation}%
By (\ref{lim1})-(\ref{lim3}),%
\begin{equation*}
\int_{S}\liminf_{n\rightarrow \infty ,w\rightarrow z}\phi _{n}(w)\mathrm{d}%
\mu (z)\leq \liminf_{m\rightarrow \infty }\int_{S}\underline{\psi }_{m}(z)%
\mathrm{d}\mu _{m}(z)\leq \liminf_{m\rightarrow \infty }\int_{S}\phi _{m}(z)%
\mathrm{d}\mu _{m}(z),
\end{equation*}%
where the last inequality is due to $\underline{\psi }_{m}(z)\leq \psi
_{m}(z)\leq \phi _{m}(z)$. 


Suppose that $\{\phi _{n}\}_{n\in 
	\mathbb{N}
}$ is uniformly bounded by $K>0$, by applying (\ref{inf}) to $\{K-\phi
_{n}\}_{n\in 
	\mathbb{N}
}$, we obtain 
\begin{align*}
K\mu (S)-\int_{S}\limsup_{n\rightarrow \infty , w\rightarrow z}\phi _{n}(w)\mathrm{d}\mu
(z)& =\int_{S}\liminf_{n\rightarrow \infty, w\rightarrow z}(K-\phi _{n}(w))\mathrm{d}\mu (z)
\\
& \leq \liminf_{n\rightarrow \infty }\int_{S}(K-\phi _{n}(z))\mathrm{d}\mu
_{n}(z) \\
& =K\liminf_{n\rightarrow \infty }\mu _{n}(S)-\limsup_{n\rightarrow \infty
}\int_{S}\phi _{n}(z)\mathrm{d}\mu _{n}(z).
\end{align*}%
Therefore, inequality (\ref{sup}) follows immediately by the weak
convergence of $\mu _{n}$ and $\mu (S)<\infty $.

\subsection{Assumption verification for Example \ref{eg PEMP}\label{ver pemp}}

\begin{lemma}
\label{lem pemp}Consider the PMEP (\ref{PEMP}) with a continuous jump size
distribution $F(\cdot)$. For $q,s\geq0$ and $0<u_{0}<x_{0}<v_{0}$, we have%
\[
\lim_{(u,v)\downarrow(u_{0},v_{0})}g(x_{0};u,v)=\lim_{(x,u)\uparrow
(x_{0},u_{0})}g(x;u,v_{0})=g(x_{0},u_{0},v_{0}),
\]
where the function $g(x;u,v)$ is any of the following three functions:
$B_{1}^{(q)}(x;u,v)$, $B_{2}^{(q,s)}(x;u,v)$ and $C^{(q,s)}(x;u,v)$.
\end{lemma}

\begin{proof}
\BIN{Note that the condition $0<u_{0}<x_{0}<v_{0}$ is to ensure the process $X$ remains positive before exiting these finite intervals, which further implies $X$ is upward regular and creeps upward.}
We limit our proof to
\begin{equation}
\lim_{(u,v)\downarrow(u_{0},v_{0})}B_{1}^{(q)}(x_{0};u,v)=B_{1}^{(q)}%
(x_{0};u_{0},v_{0}). \label{B1eq}%
\end{equation}
The other results can be proved in a similar manner.
\BIN{By the relationship $v>v_0>u>u_0$, we have}
\begin{align}
&  \left\vert B_{1}^{(q)}(x_{0};u_{0},v_{0})-B_{1}^{(q)}(x_{0};u,v)\right\vert
\nonumber\\
&  \leq\left\vert \mathbb{E}_{x_{0}}\left[  e^{-qT_{v_{0}}^{+}}1_{\{T_{v_{0}%
}^{+}<T_{u_{0}}^{-},X_{T_{v_{0}}^{+}}=v_{0}\}}\right]  -\mathbb{E}_{x_{0}%
}\left[  e^{-qT_{v}^{+}}1_{\{T_{v}^{+}<T_{u}^{-},X_{T_{v}^{+}}=v,X_{T_{v_{0}%
}^{+}}=v_{0}\}}\right]  \right\vert \nonumber\\
&  +\mathbb{P}_{x_{0}}\left\{  v_{0}<X_{T_{v_{0}}^{+}}\leq v\right\}  .
\label{B1 in}%
\end{align}
It is clear that the last term of (\ref{B1 in}) vanishes as $v\downarrow
v_{0}$ by the right-continuity of the distribution function of $X_{T_{v_{0}%
}^{+}}$. Also,%
\begin{align}
&  \left\vert \mathbb{E}_{x_{0}}\left[  e^{-qT_{v_{0}}^{+}}1_{\{T_{v_{0}}%
^{+}<T_{u_{0}}^{-},X_{T_{v_{0}}^{+}}=v_{0}\}}\right]  -\mathbb{E}_{x_{0}%
}\left[  e^{-qT_{v}^{+}}1_{\{T_{v}^{+}<T_{u}^{-},X_{T_{v}^{+}}=v,X_{T_{v_{0}%
}^{+}}=v_{0}\}}\right]  \right\vert \nonumber\\
&  =\mathbb{E}_{x_{0}}\left[  e^{-qT_{v_{0}}^{+}}1_{\{T_{v_{0}}^{+}<T_{u}%
^{-},X_{T_{v_{0}}^{+}}=v_{0}\}}\right]  -\mathbb{E}_{x_{0}}\left[
e^{-qT_{v}^{+}}1_{\{T_{v}^{+}<T_{u}^{-},X_{T_{v}^{+}}=v,X_{T_{v_{0}}^{+}%
}=v_{0}\}}\right] \nonumber\\
&  +\mathbb{E}_{x_{0}}\left[  e^{-qT_{v_{0}}^{+}}1_{\{T_{u}^{-}<T_{v_{0}}%
^{+}<T_{u_{0}}^{-},X_{T_{v_{0}}^{+}}=v_{0}\}}\right] \nonumber\\
&  \leq1-\mathbb{E}_{v_{0}}\left[  e^{-qT_{v}^{+}}1_{\{T_{v}^{+}<T_{u}%
^{-},X_{T_{v}^{+}}=v\}}\right]  +\mathbb{P}_{x_{0}}\left\{  T_{u}^{-}%
<T_{v_{0}}^{+}<T_{u_{0}}^{-}\right\}  . \label{10}%
\end{align}
Let $\zeta$ be the time of the first jump of the compound Poisson process $Z$
with jump rate $\lambda>0$.
\BIN{Note that $X$ will increase continuously up to time $\zeta$ as long as the initial value is positive. Since $v>v_0>0$, we have}%
\begin{equation}
1-\mathbb{E}_{v_{0}}\left[  e^{-qT_{v}^{+}}1_{\{T_{v}^{+}<T_{u}^{-}%
,X_{T_{v}^{+}}=v\}}\right]  \leq1-\mathbb{E}_{v_{0}}\left[  e^{-qT_{v}^{+}%
}1_{\{\textcolor[rgb]{0.0,0.0,0.0}{\zeta}>T_{v}^{+}\}}\right]  =1-\left(
\frac{v}{v_{0}}\right)  ^{-(q+\lambda)/\mu}. \label{11}%
\end{equation}
\BIN{By conditioning on $X_{T_{u}^{-}-}$, one obtains}
\begin{align}
\mathbb{P}_{x_{0}}\left\{  T_{u}^{-}<T_{v_{0}}^{+}<T_{u_{0}}^{-}\right\}   &
\leq\int_{u}^{v_{0}}\mathbb{P}_{x_{0}}\left\{  X_{T_{u}^{-}-}\in
\mathrm{d}y\right\}  \mathbb{P}\left\{  y-u<J\leq y-u_{0}\right\} \nonumber\\
&  \leq\max_{u_{0}\leq y\leq v_{0}}\left(  F(y-u_{0})-F(y-u)\right)  .
\label{12}%
\end{align}
Since $F(\cdot)$ is continuous, and hence uniformly continuous for
$y\in\lbrack0,v_{0}-u_{0}]$, it follows that the right-hand side of (\ref{12})
vanishes as $u\downarrow u_{0}$. From (\ref{B1 in})--(\ref{12}), we conclude
that (\ref{B1eq}) holds.

\BIN{Note that although (\ref{12}) only uses the continuity of $F$ on $[0,\infty)$,
the proof for $C^{(q,s)}(x;u,v)$ will use the continuity of $F$ on
$(-\infty,0]$.}\bigskip
\end{proof}

\begin{proposition}
\label{propA1} Assumptions (\textbf{A1})-(\textbf{A3}) hold for the piecewise
exponential Markov process (\ref{PEMP}) with a continuous jump size
distribution $F(\cdot)$ \BIN{and initial value $X_0\geq a$}.
\end{proposition}

\begin{proof}
For $0<u<x<v$, by the strong Markov property, we have
\begin{align}
B_{1}^{(q)}(x;u,v)  &  =\mathbb{E}_{x}\left[  e^{-qT_{v}^{+}}1_{\{T_{v}%
^{+}<T_{u}^{-},X_{T_{v}^{+}}=v,\zeta>T_{v}^{+}\}}\right]  +\mathbb{E}%
_{x}\left[  e^{-qT_{v}^{+}}1_{\{T_{v}^{+}<T_{u}^{-},X_{T_{v}^{+}}%
=v,\zeta<T_{v}^{+}\}}\right] \nonumber\\
&  =\left(  \frac{v}{x}\right)  ^{-(q+\lambda)/\mu}+\lambda\int_{0}^{\frac
{1}{\mu}\ln\frac{v}{x}}e^{-(q+\lambda)t}\mathrm{d}t\int_{u-xe^{\mu t}%
}^{v-xe^{\mu t}}B_{1}^{(q)}(xe^{\mu t}+w;u,v)F(\mathrm{d}w). \label{B1}%
\end{align}
By Lemma \ref{lem pemp}, \BIN{Eq. (\ref{B1})}, and the
dominated convergence theorem, it is straightforward to verify that Assumption
(\textbf{A1}) holds and \BIN{for $x>a$,}%
\[
b_{a,1}^{(q)}(x)=\frac{q+\lambda}{\mu x}-\frac{\lambda}{\mu x}\int_{-a}%
^{0}B_{1}^{(q)}(x+w;x-a,x)F(\mathrm{d}w).
\]
\BIN{Note that we require $x>a$ as otherwise $x+w $ in the above equation could be negative for $w\in(-a,0)$, and then Lemma \ref{lem pemp} does not apply.}
Obviously, $\int_{x}^{y}b_{a,1}^{(q)}(w)\mathrm{d}w<\infty$ for all
$0<x<y<\infty$. Similarly, by conditioning on the first jump of $Z$,
\BIN{for $0<u<x<v$,}%
\begin{align*}
B_{2}^{(q)}(x,\mathrm{d}z;u,v)  &  =\lambda\int_{0}^{\frac{1}{\mu}\ln\frac
{v}{x}}e^{-(q+\lambda)t}F(v-xe^{\mu t}+\mathrm{d}z)\mathrm{d}t\\
&  +\lambda\int_{0}^{\frac{1}{\mu}\ln\frac{v}{x}}e^{-(q+\lambda)t}%
\mathrm{d}t\int_{u-xe^{\mu t}}^{v-xe^{\mu t}}B_{2}^{(q)}(xe^{\mu
t}+w,\mathrm{d}z;u,v)F(\mathrm{d}w),
\end{align*}
and
\[
C^{(q,s)}(x;u,v)=\lambda\int_{0}^{\frac{1}{\mu}\ln\frac{v}{x}}e^{-(q+\lambda
)t}\mathrm{d}t\int_{-\infty}^{v-xe^{\mu t}}C^{(q,s)}(xe^{\mu t}%
+w;u,v)F(\mathrm{d}w),
\]
where it is understood that $C^{(q,s)}(xe^{\mu t}+w;u,v)=e^{s(xe^{\mu t}%
+w-u)}$ for $w<u-xe^{\mu t}$. One can verify from Lemma \ref{lem pemp} and the
dominated convergence theorem that Assumptions (\textbf{A2}) and (\textbf{A3})
hold, \BIN{and for $x>a,$}
\[
\textcolor[rgb]{0.0,0.0,0.0}{b_{a,2}^{(q)}(x,\mathrm{d}z)}=\frac{\lambda
}{\mu x}F(\mathrm{d}z)+\frac{\lambda}{\mu x}\int_{-a}^{0}B_{2}^{(q)}%
(x+w,\mathrm{d}z;x-a,x)F(\mathrm{d}w),
\]
and
\[
c_{a}^{(q,s)}(x)=\frac{\lambda}{\mu x}\int_{-\infty}^{0}C^{(q,s)}%
(x+w;x-a,x)F(\mathrm{d}w).
\]
This ends the proof.
\end{proof}

\subsection{Assumption verification for Example \ref{eg JD}\label{ver JD}}

Let $U$ be the continuous component of $X$, which is a linear diffusion
process with the infinitesimal generator $\mathcal{L}_{U}=\frac{1}{2}%
\sigma^{2}(y)\frac{\mathrm{d}^{2}}{\mathrm{d}y^{2}}+\mu(y)\frac{\mathrm{d}%
}{\mathrm{d}y}.$ It is well-known that, for any $q>0$, there exist two
independent and positive solutions, denoted as $\phi_{q}^{\pm}(y)$, to the
Sturm-Liouville equation
\begin{equation}
\mathcal{L}_{U}\phi_{q}^{\pm}(y)=q\phi_{q}^{\pm}(y), \label{SL}%
\end{equation}
where $\phi_{q}^{+}(\cdot)$ is strictly increasing and $\phi_{q}^{-}(\cdot)$
is strictly decreasing. By the Lipschitz assumption on $\mu(\cdot)$ and
$\sigma(\cdot)$, it follows from the Schauder estimates (e.g., Theorem 6.14 of
Gilbarg and Trudinger \cite{GT01}) of Eq. (\ref{SL}) that $\phi_{q}^{\pm
}(\cdot)\in C^{2,\alpha}(\bar{\Omega})$ for any $\alpha\in(0,1]$ and any
compact set $\bar{\Omega}\subset%
\mathbb{R}
$. Interested readers can refer to Section 4.1 of Gilbarg and Trudinger
\cite{GT01} for more detail on the H\"{o}lder space $C^{2,\alpha}(\bar{\Omega
})$.

We denote the first hitting time of $U$ to level $z\in%
\mathbb{R}
$ by $H_{z}=\inf\{t>0:U_{t}=z\}$. It is well-known that, for $u\leq x\leq v$,%
\begin{equation}
\mathbb{E}_{x}\left[  e^{-qH_{u}}1_{\left\{  H_{u}<H_{v}\right\}  }\right]
=\frac{f_{q}(x,v)}{f_{q}(u,v)}\quad\text{and\quad}\mathbb{E}_{x}\left[
e^{-qH_{v}}1_{\left\{  H_{v}<H_{u}\right\}  }\right]  =\frac{f_{q}(u,x)}%
{f_{q}(u,v)}, \label{2sd}%
\end{equation}
where $f_{q}(x,y):=\phi_{q}^{+}(x)\phi_{q}^{-}(y)-\phi_{q}^{+}(y)\phi_{q}%
^{-}(x)$. Note that $f_{q}(x,y)$ is strictly decreasing in $x$ and strictly
increasing in $y$ with $f_{q}(x,x)=0$. In particular, for $u\leq x\leq v$, we
have
\begin{equation}
\mathbb{E}_{x}\left[  {e}^{-qH_{u}}\right]  =\frac{\phi_{q}^{-}(x)}{\phi
_{q}^{-}(u)}\quad\text{and\quad}\mathbb{E}_{x}\left[  {e}^{-qH_{v}}\right]
=\frac{\phi_{q}^{+}(x)}{\phi_{q}^{+}(v)}. \label{1sd}%
\end{equation}
For $\mathbf{e}_{q}$ an independent exponential random variable with mean
$1/q<\infty$, the $q$-potential measure of $U$ is given by
\[
r_{q}(x,y):=\frac{1}{q}\mathbb{P}_{x}\left\{  U_{\mathbf{e}_{q}}\in
\mathrm{d}y\right\}  /\mathrm{d}y=\left\{
\begin{array}
[c]{l}%
\frac{2}{q\sigma^{2}(y)}\frac{\phi_{q}^{+}(x)\phi_{q}^{-}(y)}{f_{q,1}%
(y,y)},\quad x\leq y,\\
\frac{2}{q\sigma^{2}(y)}\frac{\phi_{q}^{+}(y)\phi_{q}^{-}(x)}{f_{q,1}%
(y,y)},\quad x>y,
\end{array}
\right.
\]
where $f_{q,1}(x,y):=\frac{\partial}{\partial x}f_{q}(x,y).$ Furthermore, the
$q$-potential measure of $U$ killed on exiting the interval $[u,v]$, for
$u\leq x,y\leq v$, is given by
\begin{align}
\theta^{(q)}(x,y;u,v)  &  :=\frac{1}{q}\mathbb{P}_{x}\left(  U_{\mathbf{e}%
_{q}}\in\mathrm{d}y,\mathbf{e}_{q}<H_{u}\wedge H_{v}\right)  /\mathrm{d}%
y\nonumber\\
&  =r_{q}(x,y)-\frac{f_{q}(x,v)}{f_{q}(u,v)}r_{q}(u,y)-\frac{f_{q}(u,x)}%
{f_{q}(u,v)}r_{q}(v,y). \label{theta}%
\end{align}

The next lemma is an analogy of Lemma \ref{lem pemp}. Thanks to the diffusion
term in the jump diffusion model (\ref{JD}), we now allow for the presence of
atoms in the jump intensity measure $\nu(\cdot)$.

\begin{lemma}
\label{lem JD}Consider the jump diffusion model (\ref{JD}). For $q,s\geq0$ and
$u_{0}<x_{0}<v_{0}$, we have%
\[
\lim_{(u,v)\downarrow(u_{0},v_{0})}g(x_{0};u,v)=\lim_{(x,u)\uparrow
(x_{0},u_{0})}g(x;u,v_{0})=g(x_{0},u_{0},v_{0}),
\]
where $g(x;u,v)$ is any of the following functions: $B_{1}^{(q)}(x;u,v)$,
$B_{2}^{(q,s)}(x;u,v)$ and $C^{(q,s)}(x;u,v)$.
\end{lemma}

\begin{proof}
We can follow the same proof as Lemma \ref{lem pemp} except for the term
$\mathbb{P}_{x_{0}}\left\{  T_{u}^{-}<T_{v_{0}}^{+}<T_{u_{0}}^{-}\right\}  $
in (\ref{12}), which will be handled distinctly here. We have $X_{t}=U_{t}$
a.s. for $t<\zeta$, where $\zeta$ is the first time a jump occurs which
follows an exponential distribution with mean $1/\lambda=1/\nu(\mathbb{R})>0$.
For any $u_{0}<u<x_{0}<v_{0}$, by (\ref{2sd}) and (\ref{1sd}), we have
\begin{align*}
\mathbb{P}_{x_{0}}\left\{  T_{u}^{-}<T_{v_{0}}^{+}<T_{u_{0}}^{-}\right\}   &
\leq\mathbb{P}_{u}\left\{  T_{v_{0}}^{+}<T_{u_{0}}^{-}\right\} \\
&  =\mathbb{P}_{u}\left\{  T_{v_{0}}^{+}<T_{u_{0}}^{-},\xi>T_{v_{0}}%
^{+}\right\}  +\mathbb{P}_{u}\left\{  \xi\leq T_{v_{0}}^{+}<T_{u_{0}}%
^{-}\right\} \\
&  \leq\mathbb{E}_{u}\left[  e^{-\lambda H_{v_{0}}}1_{\left\{  H_{v_{0}%
}<H_{u_{0}}\right\}  }\right]  +1-\mathbb{E}_{u}\left[  e^{-\lambda H_{u_{0}}%
}\right] \\
&  =\frac{f_{q}(u_{0},u)}{f_{q}(u_{0},v_{0})}+1-\frac{\phi_{q}^{-}(u)}%
{\phi_{q}^{-}(u_{0})}.
\end{align*}
Therefore, it follows that $\lim_{u\downarrow u_{0}}\mathbb{P}_{x_{0}}\left\{
T_{u}^{-}<T_{v_{0}}^{+}<T_{u_{0}}^{-}\right\}  =0$ by $f_{q}(u_{0},u_{0}%
)=0$.\bigskip
\end{proof}

\begin{proposition}
Assumptions (\textbf{A1})-(\textbf{A3}) hold for the jump diffusion model
(\ref{JD}).
\end{proposition}

\begin{proof}
By the strong Markov property, (\ref{2sd}) and (\ref{theta}), for $u<x<v$, it
follows that%
\begin{align*}
&  B_{1}^{(q)}(x;u,v)\\
&  =\mathbb{E}_{x}\left[  e^{-qT_{v}^{+}}1_{\{T_{v}^{+}<T_{u}^{-},T_{v}%
^{+}=v,\zeta>T_{v}^{+}\}}\right]  +\mathbb{E}_{x}\left[  e^{-qT_{v}^{+}%
}1_{\{T_{v}^{+}<T_{u}^{-},T_{v}^{+}=v,\zeta<T_{v}^{+}\}}\right] \\
&  =\mathbb{E}_{x}\left[  e^{-(q+\lambda)H_{v}}1_{\{H_{v}<H_{u}\}}\right]
+\int_{u}^{v}\mathbb{E}_{x}\left[  e^{-q\zeta}1_{\{\zeta<H_{u}\wedge
H_{v},U_{\zeta}\in\mathrm{d}y\}}\right]  \int_{\mathbb{R}}B_{1}^{(q)}%
(y+\gamma(y,w);u,v)\frac{\nu(\mathrm{d}w)}{\lambda}\\
&  =\frac{f_{q+\lambda}(u,x)}{f_{q+\lambda}(u,v)}+\int_{u}^{v}\theta
^{(q+\lambda)}(x,y;u,v)\mathrm{d}y\int_{\mathbb{R}}B_{1}^{(q)}(y+\gamma
(y,w);u,v)\nu(\mathrm{d}w),
\end{align*}
where it is understood that $B_{1}^{(q)}(y+\gamma(y,w);u,v)=0$ if
$\gamma(y,w)>v-y$ or $\gamma(y,w)<u-y$. By Lemma \ref{lem JD}, the dominated
convergence theorem, and the identity $f_{q+\lambda}(u,v)=-f_{q+\lambda}%
(v,u)$, we can verify that Assumption (\textbf{A1}) holds with
\[
b_{a,1}^{(q)}(x)=\frac{-f_{q+\lambda,1}(x-a,x)}{f_{q+\lambda}(x-a,x)}%
-\int_{x-a}^{x}\tilde{\theta}_{a}^{(q+\lambda)}(x,y)\mathrm{d}y\int
_{\mathbb{R}}B_{1}^{(q)}(y+\gamma(y,w);x-a,x)\nu(\mathrm{d}w),
\]
where we write $\tilde{\theta}_{a}^{(q+\lambda)}(x,y):=-\frac{f_{q+\lambda
,1}(x-a,x)}{f_{q+\lambda}(x-a,x)}r_{q+\lambda}(x,y)-r_{q+\lambda,1}%
(x,y)+\frac{f_{q+\lambda,1}(x,x)}{f_{q+\lambda}(x-a,x)}r_{q+\lambda}(x-a,y)$
and $r_{q+\lambda,1}(x,y):=\frac{\partial}{\partial x}r_{q+\lambda}(x,y).$ The
integrability of $b_{a,1}^{(q)}(\cdot)$ follows from the continuity of the
$\phi_{q}^{+}(\cdot)$ and $\phi_{q}^{-}(\cdot)$.

Similarly, by the strong Markov property of $X$, (\ref{2sd}) and
(\ref{theta}), we have%
\[
B_{2}^{(q)}(x,\mathrm{d}z;u,v)=\int_{u}^{v}\theta^{(q+\lambda)}%
(x,y;u,v)\mathrm{d}y\int_{\mathbb{R}}B_{2}^{(q)}(y+\gamma(y,w),\mathrm{d}%
z;u,v)\nu(\mathrm{d}w),
\]
and%
\[
C^{(q,s)}(x;u,v)=\frac{f_{q+\lambda}(x,v)}{f_{q+\lambda}(u,v)}+\int_{u}%
^{v}\theta^{(q+\lambda)}(x,y;u,v)\mathrm{d}y\int_{\mathbb{R}}C^{(q,s)}%
(y+\gamma(y,w);u,v)\nu(\mathrm{d}w).
\]
One can verify from Lemma \ref{lem JD} that Assumptions (\textbf{A2}) and
(\textbf{A3}) hold with
\[
b_{2,a}^{(q)}(x,\mathrm{d}z)=\int_{x-a}^{x}\tilde{\theta}_{a}^{(q+\lambda
)}(x,y)\mathrm{d}y\int_{\mathbb{R}}B_{2}^{(q)}(y+\gamma(y,w),\mathrm{d}%
z;x-a,x)\nu(\mathrm{d}w),
\]
and
\[
c_{a}^{(q,s)}(x)=\frac{-f_{q+\lambda,1}(x,x)}{f_{q+\lambda}(x-a,x)}+\int
_{x-a}^{x}\tilde{\theta}^{(q+\lambda)}(x,y)\mathrm{d}y\int_{\mathbb{R}%
}C^{(q,s)}(y+\gamma(y,z);x-a,x)\nu(\mathrm{d}w).
\]
This completes the proof.
\end{proof}

\end{document}